\documentclass[11pt]{article}

\usepackage[utf8]{inputenc}
\usepackage[T1]{fontenc}
\usepackage{newtxtext,newtxmath}       % Professional Times-like font
\usepackage{microtype}                 % Micro-typography for better spacing
\usepackage[margin=1in]{geometry}      % Standard consistent margins
\usepackage[dvipsnames]{xcolor}        % For NavyBlue links
\usepackage{setspace}
\usepackage{fancyhdr}
\definecolor{TitleNavy}{HTML}{113A5C}
\definecolor{TitleGold}{HTML}{9C7A2E}
\definecolor{AbstractBG}{HTML}{F5F8FB}
\definecolor{AbstractBorder}{HTML}{C7D3E0}

\usepackage{amsmath,mathtools,bm}
\makeatletter
\@ifundefined{openbox}{}{\let\openbox\@undefined}%
\makeatother
\usepackage{amsthm}
\usepackage{bbm}
\numberwithin{equation}{section}

\usepackage{graphicx}
\usepackage{booktabs}
\usepackage{array}
\usepackage{threeparttable}
\usepackage{caption}
\usepackage{algorithm}
\usepackage[noend]{algpseudocode}

\newcommand{\mcse}[1]{{\footnotesize\,(#1)}}

\usepackage{enumitem}
\usepackage{siunitx}
\setlist[itemize]{leftmargin=1.8em, itemsep=0.25em, topsep=0.35em}
\setlist[enumerate]{leftmargin=2em, itemsep=0.25em, topsep=0.35em}

\usepackage{natbib}
\setcitestyle{authoryear,round,semicolon}

\providecommand\IfDocumentMetadataT[1]{}
\usepackage[
  colorlinks=true,
  linkcolor=TitleNavy,
  citecolor=TitleNavy,
  urlcolor=TitleNavy
]{hyperref}
\usepackage[nameinlink,noabbrev]{cleveref}
\makeatletter
\renewcommand\section{\@startsection{section}{1}{\z@}%
  {-3.5ex \@plus -1ex \@minus -.2ex}%
  {2.3ex \@plus .2ex}%
  {\normalfont\Large\bfseries\color{TitleNavy}}}
\renewcommand\subsection{\@startsection{subsection}{2}{\z@}%
  {-3.0ex \@plus -1ex \@minus -.2ex}%
  {1.5ex \@plus .2ex}%
  {\normalfont\large\bfseries\color{TitleNavy}}}
\renewcommand\subsubsection{\@startsection{subsubsection}{3}{\z@}%
  {-2.5ex \@plus -1ex \@minus -.2ex}%
  {1.0ex \@plus .2ex}%
  {\normalfont\normalsize\bfseries\color{TitleNavy}}}
\makeatother

\newtheoremstyle{paperplain}
  {6pt plus 2pt minus 1pt}
  {6pt plus 2pt minus 1pt}
  {\itshape}
  {0pt}
  {\color{TitleNavy}\bfseries}
  {.}
  {0.6em}
  {}
\newtheoremstyle{paperdefinition}
  {6pt plus 2pt minus 1pt}
  {6pt plus 2pt minus 1pt}
  {\normalfont}
  {0pt}
  {\color{TitleNavy}\bfseries}
  {.}
  {0.6em}
  {}
\newtheoremstyle{paperremark}
  {6pt plus 2pt minus 1pt}
  {6pt plus 2pt minus 1pt}
  {\normalfont}
  {0pt}
  {\color{TitleNavy}\itshape}
  {.}
  {0.6em}
  {}

\theoremstyle{paperplain}
\newtheorem{theorem}{Theorem}[section]
\newtheorem{lemma}[theorem]{Lemma}
\newtheorem{proposition}[theorem]{Proposition}
\newtheorem{corollary}[theorem]{Corollary}

\theoremstyle{paperdefinition}
\newtheorem{definition}[theorem]{Definition}
\newtheorem{assumption}[theorem]{Assumption}

\theoremstyle{paperremark}
\newtheorem{remark}[theorem]{Remark}

\newcommand{\E}{\mathbb{E}}
\newcommand{\Var}{\operatorname{Var}}

\newcommand{\Pp}{\mathbb{P}}
\newcommand{\indep}{\mathrel{\perp\!\!\!\perp}}
\newcommand{\ind}[1]{\mathbbm{1}\{#1\}}
\newcommand{\cF}{\mathcal{F}}
\newcommand{\cG}{\mathcal{G}}

\newcommand{\pto}{\xrightarrow{p}}
\newcommand{\dto}{\xrightarrow{d}}

\newcommand{\pit}{\pi_t}

\newcommand{\mhat}{\widehat m}

\newcommand{\expit}{\mathrm{expit}}

\newcommand{\arxiv}[1]{arXiv:#1}
\newcommand{\PaperTitle}{Fixed-Horizon Self-Normalized Inference for Adaptive Experiments via Martingale AIPW/DML \\ with Logged Propensities}
\newcommand{\AuthorName}{Gabriel Saco}
\newcommand{\AuthorAffiliation}{Universidad del Pac\'ifico}
\newcommand{\AuthorORCID}{0009-0009-8751-4154}
\newcommand{\CodeRepositoryURL}{https://github.com/gsaco/martingale-aipw-dml}

\fancypagestyle{plain}{%
  \fancyhf{}
  \fancyfoot[C]{\small\thepage}

}

\begin{document}

% ============================================================================
% TITLE PAGE
% ============================================================================
\title{\PaperTitle}
\author{\AuthorName}
\date{}
\begingroup
\thispagestyle{plain}
\centering

\vspace*{-0.8em}
{\color{TitleNavy}\rule{\textwidth}{1.15pt}\par}
\vspace{1.3em}

{\LARGE\bfseries\color{TitleNavy}\PaperTitle\par}
\vspace{1.0em}

{\large \AuthorName\par}
\vspace{0.2em}
{\normalsize\textit{\AuthorAffiliation}\par}
\vspace{0.8em}

{\small\textsc{ORCID}: \href{https://orcid.org/\AuthorORCID}{\AuthorORCID}\par}
\vspace{0.3em}
{\small\textsc{Replication code}: \href{\CodeRepositoryURL}{\nolinkurl{\CodeRepositoryURL}}\par}

\vspace{0.9em}
{\color{TitleGold}\rule{0.72\textwidth}{0.9pt}\par}
\endgroup

\vspace{1.2em}
\begin{center}
\setlength{\fboxsep}{11pt}
\fcolorbox{AbstractBorder}{AbstractBG}{%
\begin{minipage}{0.93\textwidth}
\textbf{\large Abstract}
\vspace{0.45em}

\small
Adaptive randomized experiments update treatment probabilities as data accrue, but still require an end-of-study interval for the average treatment effect (ATE) at a prespecified horizon. Under adaptive assignment, propensities can keep changing, so the predictable quadratic variation of AIPW/DML score increments may remain random. When no deterministic variance limit exists, Wald statistics normalized by a single long-run variance target can be conditionally miscalibrated given the realized variance regime. We assume no interference, sequential randomization, i.i.d.\ arrivals, and executed overlap on a prespecified scored set, and we require two auditable pipeline conditions: the platform logs the executed randomization probability for each unit, and the nuisance regressions used to score unit $t$ are constructed predictably from past data only. These conditions make the centered AIPW/DML scores an exact martingale difference sequence. Using self-normalized martingale limit theory, we show that the Studentized statistic, with variance estimated by realized quadratic variation, is asymptotically $\mathcal{N}(0,1)$ at the prespecified horizon, even without variance stabilization. Simulations validate the theory and highlight when standard fixed-variance Wald reporting fails.  
\end{minipage}}
\end{center}

\newpage

% ============================================================================
% 1. INTRODUCTION
% ============================================================================

\section{Introduction}\label{sec:intro}

Adaptive randomized experiments---including response-adaptive clinical trials, contextual bandits, and large-scale platform experimentation systems---update assignment probabilities as data accrue to balance learning and deployment \citep{kasySautmann2021}. In practice, however, many platforms still require conventional end-of-study reporting for classical causal estimands such as the superpopulation average treatment effect (ATE) computed once at a prespecified horizon. Hereafter, we denote by $\pi_t$ the executed assignment probability used to randomize unit $t$ after applying any platform guardrails, as recorded in the experiment log. In this paper, we study fixed-horizon Wald inference for the standard logged-propensity AIPW/DML estimator, the sample average of doubly robust pseudo-outcomes scored using these logged propensities. Under adaptive assignment, the propensity process $\{\pi_t\}$ is itself data-dependent, so the predictable quadratic variation of the AIPW/DML score increments can remain replication-random and need not converge to a single deterministic long-run variance target.

Some end-of-study Wald arguments for AIPW/A2IPW (and related DML estimators) under adaptivity proceed via Slutsky steps that rely on a deterministic variance target for the predictable quadratic variation, often enforced through stabilization-type or design-stability conditions on assignment probabilities and/or average conditional variances \citep{hadad2021confidence,zhan2021off,katoIshiharaHondaNarita2020,cook2024semiparametric,LiOwen2024,sengupta2025designstability}. On modern platforms, however, the policy can keep reacting to noisy intermediate estimates. Clipping and guardrails can activate intermittently and batch updates can induce regime switches. When this happens, a Wald statistic normalized using a single deterministic variance target can be systematically miscalibrated conditional on the realized variance regime, even if marginal coverage appears close to nominal.

We treat the adaptive assignment policy as given but assume that the platform logs the executed propensity used to randomize each unit and that nuisance regressions used for AIPW/DML scoring are fit predictably using only past data. Related work emphasizes that careful use of the logging policy and past-only fitting is central for post-adaptive inference \citep{bibaut2021post,katoYasuiMcAlinn2021,cook2024semiparametric}. Under these auditable conditions, the centered score increments form an exact martingale difference sequence, and we obtain fixed-horizon Wald inference by studentizing with realized quadratic variation along the realized propensity path. This yields asymptotic $\mathcal{N}(0,1)$ calibration without requiring the predictable quadratic variation to converge to a deterministic long-run variance target.\footnote{We do not claim anytime-valid or optional-stopping guarantees.}

\vspace{0.5cm}
\noindent\textbf{Contributions.}
\begin{itemize}
\item \textbf{Auditable martingale scoring.} We formalize a logging/predictability contract, logged executed propensities and predictable nuisance fitting, under which centered AIPW/DML score increments form an exact martingale difference sequence (Lemma~\ref{lem:mds}).
\item \textbf{Fixed-horizon self-normalized Wald inference.} We prove that the usual Studentized statistic, with variance estimated by realized quadratic variation, is asymptotically $\mathcal{N}(0,1)$ at a prespecified horizon even when no deterministic long-run variance limit exists (Theorem~\ref{thm:studentized}).
\item \textbf{Feasible studentization.} We show that the standard plug-in studentizer used in practice consistently estimates realized quadratic variation, so the feasible Wald interval inherits the same fixed-horizon validity (Proposition~\ref{prop:feasibility}).
\item \textbf{Oracle benchmarking and nuisance-learning effects.} We provide a conditional second-moment decomposition yielding an oracle precision benchmark and isolate a nonnegative augmentation term capturing variance inflation from nuisance error (Proposition~\ref{prop:var_decomp}). Under weighted $L^2$ convergence, the feasible statistic is asymptotically oracle-equivalent (Theorem~\ref{thm:oracle_equivalence}).
\end{itemize}

The subsequent sections are structured as follows. Section~\ref{sec:related} reviews related work. Section~\ref{sec:model} states the model and assumptions. Section~\ref{sec:estimation} presents the estimator and auditable implementation details. Section~\ref{sec:main_results} develops the main theoretical results and Section~\ref{sec:simulation} reports simulations. The appendices collect supporting limit-theory background, additional results and proofs, and an operational logging protocol.

\section{Related Work}\label{sec:related}

\subsection*{Inference after adaptive data collection}
A growing literature studies inference with adaptively collected data, where observations are generated under an evolving information set and classical i.i.d.\ arguments do not directly apply. Early econometric work by \citet{hahnHiranoKarlan2011} highlighted how propensity information can be leveraged for inference in sequential designs. More recent general frameworks derive asymptotic representations for sequential decisions and adaptive experiments under broad conditions \citep{hiranoPorter2023asymptotics}. In the contextual-bandit and adaptive-experiment literature, fixed-horizon inference has also been developed via batched OLS/batchwise studentization arguments \citep{zhang2020large}. Our focus is narrower but operationally central for experimentation platforms: fixed-horizon ATE reporting with logged executed propensities and predictable AIPW/DML scoring. In particular, we target settings where the predictable quadratic variation remains random across replications, so that deterministic-variance normalizations can be conditionally miscalibrated. General background on response-adaptive randomization and bandit-style designs can be found in \citet{rosenberger2015randomization} and \citet{villar2015bandit}.

\subsection*{Stabilization, adaptive weighting and batching}
In adaptive experimentation and off-policy evaluation, evolving propensities can create heavy tails and regime-dependent uncertainty, motivating variance-control strategies. In policy evaluation, \citet{hadad2021confidence} and \citet{zhan2021off} develop adaptive weighting schemes for augmented IPW/DR scores to obtain asymptotically normal $t$-statistics. For post-contextual-bandit inference, \citet{bibaut2021post} proposes stabilized doubly robust constructions that estimate conditional scale components using only past data. Another route restricts the data-collection design---for example through batching---to recover classical CLTs for bandit estimators \citep{zhang2020large}. In contrast, we keep the standard logged-propensity AIPW/DML estimator unchanged and do not require batching or adaptive reweighting. Instead, we rely on a martingale score representation and studentization by realized quadratic variation.

\subsection*{Design stability and deterministic variance-limit CLTs}
A further line of work derives fixed-horizon CLTs for IPW/AIPW-type ATE estimators under explicit design stability conditions ensuring that inverse-propensity averages and/or average conditional variances converge to nonrandom limits, yielding conventional deterministic asymptotic variances for Wald reporting \citep{sengupta2025designstability}. Related perspectives arise when one engineers assignment rules (or batchwise designs) to target efficiency or precision for A2IPW/DML-style estimators \citep{katoIshiharaHondaNarita2020,LiOwen2024,cook2024semiparametric}. More generally, the adaptive-experiment literature emphasizes subtleties around what ``the logging policy'' means operationally when platforms implement clipping, guardrails, or algorithmic randomness \citep{katoYasuiMcAlinn2021}. Our contribution is complementary: we assume the executed propensity actually used to randomize each unit is logged and we avoid requiring convergence of the predictable quadratic variation to a deterministic limit by using realized quadratic variation as the normalizer.

\subsection*{Self-normalized martingale theory for quadratic-variation studentization}
Our fixed-horizon Wald statistic is a self-normalized martingale functional. Classic references for martingale CLTs and self-normalized processes include \citet{hallHeyde1980} and \citet{delapena2009self}, as well as the survey \citet{shao2013survey}. Modern probability theory provides refined asymptotic and nonasymptotic control for self-normalized martingales, including Berry--Esseen bounds \citep{fanShao2017}, Cram{\'e}r-type moderate deviations \citep{fanGramaLiuShao2019}, and concentration inequalities \citep{bercuTouati2019}. We apply this theory to AIPW/DML score increments under adaptive assignment, using realized quadratic variation to obtain a fixed-horizon $\mathcal{N}(0,1)$ approximation without a deterministic variance limit.

\subsection*{Anytime-valid and time-uniform alternatives}
The results in this paper are prespecified-horizon and do not provide optional-stopping guarantees. When inference must remain valid under continuous monitoring or data-dependent stopping, time-uniform methods based on test supermartingales and confidence sequences are appropriate \citep{howard2021confidence,waudby2024timeuniform}. Related time-uniform tools also appear alongside fixed-time inference in adaptive-experiment work \citep{cook2024semiparametric,katoIshiharaHondaNarita2020,waudbySmithRamdas2024betting}. We focus instead on conventional fixed-horizon reporting with realized quadratic-variation studentization (see Remark~\ref{rem:optional_stopping} for additional discussion).
\section{Model and Assumptions}\label{sec:model}

\subsection{Propensities and Logged Assignment Rule}\label{sec:propensities}

Consider a stream of experimental units indexed by $t=1,\dots,n$. For each unit, covariates $X_t$ are observed before assignment. The platform computes an assignment probability $\pi_t\in(0,1)$. Then, randomizes $A_t\in\{0,1\}$ using $\pi_t$. Finally, observes an outcome $Y_t$. 

We observe and store $Z_t:=(X_t,A_t,Y_t,\pi_t)$ in time order. Let
\begin{equation}\label{eq:Ft}
\cF_t := \sigma(Z_1,\dots,Z_t), \qquad \cF_0 = \{\emptyset,\Omega\},
\end{equation}
be the data filtration. Let
\begin{equation}\label{eq:Gt}
\cG_t := \sigma(\cF_{t-1},X_t,\pi_t)
\end{equation}
be the $\sigma$-field immediately before randomization at time $t$, after the platform has computed the executed assignment probability $\pi_t$. The policy may be arbitrarily adaptive or randomized. The fixed-horizon validity results below require only that the realized executed propensity used to randomize $A_t$ is recorded in the log (Assumption~\ref{ass:logging_integrity}).

\begin{definition}[Adaptive logged executed propensity]\label{def:propensity}
At time $t$, after observing $(\cF_{t-1},X_t)$ and any algorithmic randomness used to form the propensity, the platform records the executed assignment probability $\pi_t\in(0,1)$ and draws treatment according to
\[
A_t \mid \cG_t \sim \mathrm{Bernoulli}(\pi_t),
\]
where $\cG_t$ is the pre-treatment $\sigma$-field in \eqref{eq:Gt}. If $\pi_t$ is deterministic given $(\cF_{t-1},X_t)$, then $\pi_t=\Pp(A_t=1\mid \cF_{t-1},X_t)$. If $\pi_t$ is produced using additional exogenous randomness, then $\Pp(A_t=1\mid \cF_{t-1},X_t)=\E[\pi_t\mid \cF_{t-1},X_t]$ while still $\Pp(A_t=1\mid \cG_t)=\pi_t$. We treat any such platform-side randomness as exogenous, i.e., independent of $(Y_t(0),Y_t(1))$ conditional on $(\cF_{t-1},X_t)$.
\end{definition}

\begin{assumption}[Logging integrity]\label{ass:logging_integrity}
For each $t$, the logged propensity $\pi_t$ is $\cG_t$-measurable, takes values in $(0,1)$, and equals the probability passed to the randomization device that generated $A_t$, i.e.,
\[
\Pp(A_t=1\mid \cG_t)=\pi_t \quad \text{a.s.}
\]
Equivalently, $A_t\mid \cG_t\sim \mathrm{Bernoulli}(\pi_t)$ with the logged $\pi_t$.
\end{assumption}

\subsection{Audit diagnostics and data contract}\label{sec:audit}

\begin{remark}[Auditing logged propensities]\label{rem:logging_audit}
Assumption~\ref{ass:logging_integrity} requires that, conditional on the platform history $\cG_t$ (including the logged probability $\pi_t$), the treatment assignment satisfies $A_t\mid \cG_t \sim \mathrm{Bernoulli}(\pi_t)$. This is a design-stage property. It holds only if the log records the probability that was actually passed to the randomization device.
A basic diagnostic is calibration of $A_t$ against $\pi_t$. For any bin $B\subset(0,1)$ with many observations, let
\[
N_B:=\sum_{t=1}^n \ind{\pi_t\in B},\qquad
\overline A_B:=\frac{1}{N_B}\sum_{t:\pi_t\in B} A_t,\qquad
\overline \pi_B:=\frac{1}{N_B}\sum_{t:\pi_t\in B} \pi_t.
\]
In practice, restrict attention to bins with $N_B\ge N_{\min}$ for a user-chosen $N_{\min}\ge 1$ (or merge empty/small bins).
Under correct logging, $\overline A_B$ should be close to $\overline \pi_B$ up to martingale sampling variability; under adaptivity, dependence can matter, so formal testing can be based on martingale concentration/self-normalized methods for the MDS $(A_t-\pi_t)$. These checks can detect severe mis-logging or implementation bugs, but they cannot certify correct logging or validate the causal assumptions in Assumptions~\ref{ass:sutva}--\ref{ass:unconf}.
\end{remark}

\begin{remark}[Timing and measurability]\label{rem:timing}
The experiment unfolds in the order (1) observe $X_t$; (2) the platform computes and logs an executed propensity $\pi_t$ as a function of $(\cF_{t-1},X_t)$; (3) treatment is drawn according to $A_t\mid \cG_t\sim\mathrm{Bernoulli}(\pi_t)$; and (4) the outcome $Y_t$ is revealed and logged. In particular, $\pi_t$ is $\cG_t$-measurable, while $(A_t,Y_t)$ are $\cF_t$-measurable.
This timing underlies the predictability requirement in Section~\ref{sec:estimation} and the data-contract protocol in Appendix~\ref{app:repro}.
\end{remark}

\begin{table}[!htbp]
\centering
\small
\begin{tabular}{@{}p{0.05\textwidth}p{0.25\textwidth}p{0.55\textwidth}@{}}
\toprule
Step & Object & Measurability / check \\
\midrule
1 & Observe covariates $X_t$ & observed at $t$ (pre-treatment) \\
2 & Choose/log $\pi_t$ & $\pi_t$ must be $\cG_t$-measurable; log the executed random-number generator (RNG) probability \\
3 & Randomize $A_t$ & $A_t\mid\cG_t\sim \mathrm{Bernoulli}(\pi_t)$ (Assumption~\ref{ass:logging_integrity}) \\
4 & Observe outcome $Y_t$ & realized after $A_t$ \\
5 & Fit/update nuisances for $t{+}1$ & nuisance used at $t{+}1$ must be $\cF_t$-measurable \\
\bottomrule
\end{tabular}
\caption{Timing and measurability. The key analytical requirement is that the nuisance used to score unit $t$ is $\cF_{t-1}$-measurable (Assumption~\ref{ass:predictable_nuis}).}
\label{tab:timing}
\end{table}

\begin{table}[!htbp]
\centering
\small
\begin{tabular}{@{}ll@{}}
\toprule
Symbol & Meaning \\
\midrule
$n$, $t$ & fixed horizon; unit index $t=1,\dots,n$ \\
$Z_t$ & logged record $(X_t,A_t,Y_t,\pi_t)$ \\
$\cF_t$, $\cG_t$ & post-outcome history; pre-treatment info \eqref{eq:Ft}--\eqref{eq:Gt} \\
$\pi_t$ & logged executed propensity; $\Pp(A_t=1\mid \cG_t)=\pi_t$ (Ass.~\ref{ass:logging_integrity}) \\
$Y_t(a)$, $Y_t$ & potential outcomes; observed $Y_t=A_tY_t(1)+(1-A_t)Y_t(0)$ \\
$m_a^\star$ & true regression $x\mapsto\E[Y(a)\mid X=x]$ \\
$\widehat m_{t-1,a}$ & predictable nuisance used at time $t$ (Ass.~\ref{ass:predictable_nuis}) \\
$\phi_t(m,\pi)$ & AIPW/DR pseudo-outcome evaluated at propensity $\pi$ \eqref{eq:phi_def} \\
$\widehat{\phi}_t$, $\widehat{\theta}$ & scored pseudo-outcome $\widehat{\phi}_t=\phi_t(\widehat m_{t-1},\pi_t)$ and estimator \eqref{eq:phi_hat} \\
$\theta_0$, $\widehat\theta$ & ATE \eqref{eq:theta0}; estimator \eqref{eq:phi_hat} \\
$\xi_t$, $S_{\mathcal{T}}$ & increment $\xi_t=\widehat\phi_t-\theta_0$; sum $S_{\mathcal{T}}=\sum_{t\in\mathcal{T}}\xi_t$ \\
$\mathcal{T}$, $n_{\mathrm{eff}}$ & scored index set; $n_{\mathrm{eff}}:=|\mathcal{T}|$ \\
$Q_{\mathcal{T}}$, $V_{\mathcal{T}}^2$, $\widehat V$ & realized/predictable QV \eqref{eq:quadratic_variations}; studentizer \eqref{eq:Vhat} \\
\bottomrule
\end{tabular}
\caption{Summary of notation.}
\label{tab:notation}
\end{table}

\subsection{Potential Outcomes and Target Parameter}\label{sec:estimand}

We adopt the potential outcomes framework with the usual no interference and consistency conventions.

\begin{assumption}[Consistency and no interference; SUTVA]\label{ass:sutva}
For each unit $t$, there are well-defined potential outcomes $(Y_t(0),Y_t(1))$ that depend only on the assignment of unit $t$ itself \citep[see][]{rubin1980comment}. The observed outcome satisfies
\[
Y_t = Y_t(A_t) = A_t Y_t(1) + (1-A_t)Y_t(0)\qquad\text{a.s.}
\]
\end{assumption}

The target parameter is the average treatment effect (ATE) in the superpopulation:
\begin{equation}\label{eq:theta0}
\theta_0 := \E\big[ Y(1)-Y(0)\big],
\end{equation}
where $(X,Y(0),Y(1))$ denotes a generic draw from the superpopulation described in Assumption~\ref{ass:iid}. Assumption~\ref{ass:sutva} is standard. It rules out spillovers or network effects, which require separate methods.

\subsection{Assumptions}\label{sec:assumptions}

We maintain the following assumptions throughout.

\begin{assumption}[Superpopulation arrivals; no selection]\label{ass:iid}
The sequence $\{(X_t,Y_t(0),Y_t(1))\}_{t=1}^n$ is independent across $t$ and identically distributed with generic draw $(X,Y(0),Y(1))$. In particular, the adaptive policy affects only treatment assignment, not which units/covariates arrive.
\end{assumption}

\noindent
We write all expectations with respect to this superpopulation draw, so $\theta_0=\E[Y(1)-Y(0)]$ and $\tau(x):=\E[Y(1)-Y(0)\mid X=x]$ satisfy $\theta_0=\E[\tau(X)]$. As noted in Remark~\ref{rem:weaker_iid}, the proofs only require the conditional mean-stationarity condition $\E[\tau(X_t)\mid \cF_{t-1}]=\theta_0$.

\begin{assumption}[Sequential randomization/no anticipation]\label{ass:unconf}\label{ass:seq_rand}
$A_t \indep (Y_t(0),Y_t(1))\mid \cG_t$, where $\cG_t$ is the pre-treatment $\sigma$-field in \eqref{eq:Gt}. Equivalently, conditional on the pre-treatment information (including the realized executed propensity $\pi_t$), the treatment draw is independent of the potential outcomes.
\end{assumption}

\begin{assumption}[Executed overlap on scored units]\label{ass:overlap}
There exists $\varepsilon>0$ such that for every time index whose score is included in the final estimator (i.e., $t\in\mathcal{T}$ in Section~\ref{sec:estimation}),
\begin{equation}\label{eq:overlap}
\varepsilon \le \pi_t \le 1-\varepsilon \qquad \text{a.s.}
\end{equation}
A sufficient operational mechanism is to enforce clipping of the executed propensity before randomization on all scored units.
\end{assumption}

\begin{assumption}[Moment conditions]\label{ass:moments}
For each $a\in\{0,1\}$, $\E[|Y(a)|^4]<\infty$.
\end{assumption}

\begin{assumption}[Nuisance stability]\label{ass:nuis_stab}
There exists a finite constant $C<\infty$ such that for each $a\in\{0,1\}$,
\[
\sup_{n\ge 1}\sup_{t\in\mathcal{T}} \E[|\widehat m_{t-1,a}(X_t)|^4]\le C.
\]
In unbounded-outcome settings, a sufficient alternative is to truncate/clamp $\widehat m_{t-1,a}(X_t)$ at a large threshold; see Remark~\ref{rem:enforce_nuis}.
\end{assumption}

\begin{remark}[How to enforce Assumption~\ref{ass:nuis_stab}]\label{rem:enforce_nuis}
If outcomes are known to be bounded, say $Y_t\in[-B,B]$, we can enforce Assumption~\ref{ass:nuis_stab} by clipping $\widehat m_{t-1,a}(x)$ to $[-B,B]$ (or slightly wider) for each $a$. For unbounded outcomes, one may instead truncate at a large threshold or use robust regression; this is a technical device to control moments and does not change the target estimand.
\end{remark}

\begin{remark}[Weaker moment conditions are possible but are not needed for this paper]\label{rem:moment_conditions}
Self-normalized martingale CLTs only require finite $(2+\delta)$ moments of the
score increments. We impose fourth moments to keep the quadratic-variation
equivalence proof fully transparent.
\end{remark}

\begin{remark}[Design-stage, analysis-stage and auditable conditions]\label{rem:assumptions}
Our validity results combine (i) substantive causal/model assumptions and (ii) pipeline conditions that can be checked from logs and stored artifacts.
\begin{enumerate}
\item \emph{Substantive assumptions.} SUTVA/no interference (Assumption~\ref{ass:sutva}) and sequential randomization with a stable target estimand (Assumption~\ref{ass:unconf}), together with i.i.d.\ arrivals/no selection (Assumption~\ref{ass:iid}) or an explicit mean-stationarity alternative (Remark~\ref{rem:weaker_iid}).
\item \emph{Auditable design-time requirements.} Logged executed propensities (Assumption~\ref{ass:logging_integrity}) and executed overlap on the scored units (Assumption~\ref{ass:overlap}), which must be enforced at assignment time if needed (e.g., by clipping $\pi_t$).
\item \emph{Auditable analysis-time requirements.} A prespecified (deterministic) scored set $\mathcal{T}$ (Assumption~\ref{ass:T_det}) and predictable nuisance construction (Assumption~\ref{ass:predictable_nuis}), which rule out ``peeking'' when fitting the regressions used to score each unit.
\item \emph{Regularity conditions.} Moment and stability bounds (Assumptions~\ref{ass:moments} and \ref{ass:nuis_stab}) and variance growth (Assumption~\ref{ass:var_growth}) needed to apply martingale limit theory.
\end{enumerate}
Finally, all confidence intervals are fixed-horizon: they should be reported only at the prespecified sample size and are generally invalid under optional stopping (Remark~\ref{rem:optional_stopping} and Appendix~\ref{app:repro}).
\end{remark}

% ============================================================================
% 3. PREDICTABLE AIPW/DML ESTIMATION
% ============================================================================
\section{Predictable AIPW/DML Estimation}\label{sec:estimation}

\subsection{The Doubly Robust Score}\label{sec:drscore}

Let $m_a^\star(x):=\E[Y(a)\mid X=x]$ and write $m^\star:=(m_0^\star,m_1^\star)$. For any candidate $m=(m_0,m_1)$ and any $\pi\in(0,1)$, define the usual augmented inverse-propensity weighted (AIPW) / doubly robust (DR) score \citep{robinsRotnitzkyZhao1994,bangRobins2005} (pseudo-outcome)
\begin{equation}\label{eq:phi_def}
\phi_t(m,\pi)
:=
\frac{A_t}{\pi}\big(Y_t-m_1(X_t)\big)
-\frac{1-A_t}{1-\pi}\big(Y_t-m_0(X_t)\big)
+m_1(X_t)-m_0(X_t).
\end{equation}
When evaluating at the logged executed propensity, we write $\phi_t(m):=\phi_t(m,\pi_t)$ for brevity. The quantity $\phi_t(m,\pi_t)$ is observable given $(X_t,A_t,Y_t,\pi_t)$ and a supplied regression pair $m$.

For the oracle regression $m^\star$, one can decompose
\begin{equation}\label{eq:phi_oracle}
\phi_t(m^\star)
=
\tau(X_t)
+\frac{A_t}{\pi_t}\big(Y_t(1)-m_1^\star(X_t)\big)
-\frac{1-A_t}{1-\pi_t}\big(Y_t(0)-m_0^\star(X_t)\big).
\end{equation}
The centered increment $\phi_t(m^\star)-\theta_0$ has mean zero and is the object governed by the martingale CLT.

\begin{remark}[Terminology]\label{rem:dr_terminology}
The estimating function $\phi_t(m,\pi_t)$ in \eqref{eq:phi_def} has the classical doubly robust/AIPW form. When evaluated at the true regression functions $m^\star$ (and the executed propensities), the centered score $\phi_t(m^\star,\pi_t)-\theta_0$ coincides with the efficient influence function for the ATE in the corresponding i.i.d.\ model with known propensity (see, e.g., \citet{hahn1998role,tsiatis2006semiparametric}). For general $m$, it is simply an estimating function with the same doubly robust form. In off-policy evaluation and contextual-bandit settings, essentially the same form appears as a doubly robust score for value estimation \citep{dudikLangfordLi2011,dudikErhanLangfordLi2014}.
\end{remark}

\begin{remark}[Multi-arm extensions]\label{rem:multiarm}
The results extend directly to $K>2$ arms by replacing the scalar propensity with a probability vector and using the corresponding multivariate AIPW score. For simplicity, we present the binary case.
\end{remark}

\subsection{Forward Cross-Fitting}\label{sec:fcf}

Standard cross-fitting partitions data into folds and estimates nuisances on held-out folds. Under adaptivity, we must respect time: nuisance estimates used at time $t$ must be constructed from data strictly before $t$.

\begin{definition}[Forward cross-fitting]\label{def:forward}
Partition indices $\{1, \ldots, n\}$ into $K$ contiguous blocks $I_1, \ldots, I_K$. The partition is deterministic and fixed prior to data collection. For each block $k \ge 2$:
\begin{enumerate}[leftmargin=2em, label=\arabic*.]
\item Estimate nuisance functions $(\mhat^{(-k)}_0, \mhat^{(-k)}_1)$ using only data from blocks $I_1, \ldots, I_{k-1}$.
\item For each $t \in I_k$, evaluate the score using these estimates and the logged propensity $\pit$.
\end{enumerate}
For $t \in I_1$, use a pilot estimate or exclude $I_1$ from inference.
\end{definition}

\begin{remark}[Common pitfall]\label{rem:predictability_cf}
A common pitfall is to use i.i.d.-style sample splitting or cross-fitting patterns that inadvertently allow information from unit $t$ (or future units) to enter the nuisance regression used to score unit $t$. In adaptive experiments, such ``leakage'' violates Assumption~\ref{ass:predictable_nuis} and can break the martingale difference property in Lemma~\ref{lem:mds}. As a result, the Studentized Wald statistic built from leaky scores generally lacks the fixed-horizon guarantee of Theorem~\ref{thm:studentized}: it may still work in some finite-sample designs, but its validity is no longer ensured by the martingale argument. Appendix~\ref{app:repro} gives an auditable implementation pattern (forward cross-fitting) that enforces predictability.
\end{remark}

\begin{assumption}[Predictable nuisance construction]\label{ass:predictable_nuis}
For each scored time $t\in\mathcal{T}$ and treatment arm $a\in\{0,1\}$, the regression function $\widehat m_{t-1,a}$ used in the score for unit $t$ is $\mathcal{F}_{t-1}$-measurable. Equivalently, conditional on the past $\mathcal{F}_{t-1}$, the function $\widehat m_{t-1,a}$ does not depend on $(A_t,Y_t)$ or on any future data.
\end{assumption}

\noindent
Assumption~\ref{ass:predictable_nuis} is the formal predictability requirement (``no-peeking'', i.e., no use of contemporaneous or future outcomes) that makes the scored pseudo-outcomes a martingale difference sequence, and predictable fits are not based on non-predictable (``leaky'') scores (see Lemma~\ref{lem:operational_pred} below). The measurability here refers to the fitted function object $\widehat m_{t-1,a}$; it may then be evaluated at the current covariate $X_t$ to form $\widehat m_{t-1,a}(X_t)$. It rules out i.i.d.-style cross-fitting schemes that, even indirectly, use contemporaneous or future outcomes when constructing the nuisance for time $t$. A simple way to enforce predictability is forward cross-fitting (Definition~\ref{def:forward}): partition the time axis into blocks $I_1,\dots,I_K$, fit each nuisance model once per block using only data from $\cup_{j<k} I_j$, and hold the fitted model fixed while scoring units in $I_k$. To make this requirement auditable, an implementation should persist (i) the training indices used for each fit, (ii) fold assignments if internal cross-validation is used, and (iii) all learner randomness.

\begin{lemma}[A sufficient operational condition for predictability]\label{lem:operational_pred}
Under forward cross-fitting (Definition~\ref{def:forward}), if for each block
$I_k$ the analyst fits $(\widehat m_0^{(-k)},\widehat m_1^{(-k)})$ using only
data from blocks $I_1,\dots,I_{k-1}$ and then reuses these fitted objects
unchanged for all $t\in I_k$, then Assumption~\ref{ass:predictable_nuis} holds.
\end{lemma}

\begin{proof}
For $t\in I_k$, the fitted objects depend only on
$\sigma(Z_s: s\in I_1\cup\cdots\cup I_{k-1})\subseteq \mathcal{F}_{t-1}$ and are
therefore $\mathcal{F}_{t-1}$-measurable.
\end{proof}

\begin{remark}[Connection to A2IPW terminology]\label{rem:a2ipw}
Our $\widehat{\phi}_t$ is the same AIPW/DR pseudo-outcome used in the adaptive-experiment literature under the name \emph{A2IPW} (adaptive AIPW). At time $t$ the score is computed using outcome regressions fitted on past data only, together with the logged executed propensity used to randomize $A_t$ \citep{katoIshiharaHondaNarita2020,cook2024semiparametric}. The present paper emphasizes that this predictability requirement is not merely a convenience: it is the condition that yields an exact martingale difference sequence and enables fixed-horizon Studentized inference without any stabilization of propensities or conditional variances.
\end{remark}

\subsection{The Estimator}\label{sec:estimator}

Let $\mathcal{T}\subseteq\{1,\dots,n\}$ denote the set of indices for which the nuisance used at time $t$ is predictable (i.e., $\widehat m_{t-1}$ is $\cF_{t-1}$-measurable).
Under forward cross-fitting with a burn-in block $I_1$, one typically takes $\mathcal{T}:=\{1,\dots,n\}\setminus I_1$.
Let $n_{\mathrm{eff}}:=\sum_{t=1}^n \ind{t\in\mathcal{T}}=|\mathcal{T}|$. Throughout the asymptotic theory we treat the scored index set $\mathcal{T}$ as
deterministic, as is the case under the
forward-block construction in Definition~\ref{def:forward}.

\begin{assumption}[Deterministic scored set]\label{ass:T_det}
The scored index set $\mathcal{T}\subseteq\{1,\dots,n\}$ is fixed prior to data collection (equivalently, $T_t:=\ind{t\in\mathcal{T}}$ is nonrandom for each $t$).
\end{assumption}

\begin{remark}[Predictable scored sets]\label{rem:T_predictable}
If $T_t:=\ind{t\in\mathcal{T}}$ is allowed to be $\cF_{t-1}$-measurable, the same proof strategy applies to $\tilde\xi_t:=T_t(\widehat\phi_t-\theta_0)$. In this case, interpret $n_{\mathrm{eff}}:=\sum_{t=1}^n T_t$ and define $Q_{\mathcal{T}}$ and $V_{\mathcal{T}}^2$ using $\tilde\xi_t$. We omit further details of the predictable-$\mathcal{T}$ extension to keep the note focused.
\end{remark}

\paragraph{Estimator and studentizing factor.}
The cross-fitted pseudo-outcome is
\begin{equation}\label{eq:phi_hat}
\widehat{\phi}_t := \phi_t(\widehat m_{t-1},\pi_t),
\qquad
\widehat{\theta} := \frac{1}{n_{\mathrm{eff}}}\sum_{t\in \mathcal{T}} \widehat{\phi}_t.
\end{equation}
where $\widehat m_{t-1}$ is the predictable nuisance estimate.
\begin{equation}\label{eq:Vhat}
\widehat V := \frac{1}{n_{\mathrm{eff}}-1}\sum_{t\in\mathcal{T}}(\widehat{\phi}_t-\widehat{\theta})^2,
\qquad\text{defined for } n_{\mathrm{eff}}\ge 2.
\end{equation}
\begin{remark}[The studentizer as realized quadratic variation]\label{rem:studentizer_qv}
Let $n_{\mathrm{eff}}:=|\mathcal{T}|$ and recall $\widehat V$ in \eqref{eq:Vhat} and $Q_{\mathcal{T}}$ in \eqref{eq:quadratic_variations}. A direct calculation gives
\begin{equation}\label{eq:vhat_qT}
(n_{\mathrm{eff}}-1)\widehat V
=
Q_{\mathcal{T}} - n_{\mathrm{eff}}(\widehat\theta-\theta_0)^2.
\end{equation}
Thus $\widehat V$ is essentially the realized quadratic variation $Q_{\mathcal{T}}$ up to the negligible centering term $n_{\mathrm{eff}}(\widehat\theta-\theta_0)^2/(n_{\mathrm{eff}}-1)$. Proposition~\ref{prop:feasibility} records a minimal condition under which $(n_{\mathrm{eff}}-1)\widehat V/Q_{\mathcal{T}}\to_p 1$.
\end{remark}

\begin{proposition}[Feasibility of the studentizer]\label{prop:feasibility}
Let $\xi_t:=\widehat\phi_t-\theta_0$ and $S_{\mathcal{T}}:=\sum_{t\in\mathcal{T}}\xi_t$, so that $\widehat\theta-\theta_0=S_{\mathcal{T}}/n_{\mathrm{eff}}$ and $Q_{\mathcal{T}}=\sum_{t\in\mathcal{T}}\xi_t^2$. If $n_{\mathrm{eff}}\to\infty$ and $S_{\mathcal{T}}/\sqrt{Q_{\mathcal{T}}}=O_p(1)$, then
\[
\frac{(n_{\mathrm{eff}}-1)\widehat V}{Q_{\mathcal{T}}}\to_p 1.
\]
In particular, the condition $S_{\mathcal{T}}/\sqrt{Q_{\mathcal{T}}}=O_p(1)$ holds whenever $S_{\mathcal{T}}/\sqrt{Q_{\mathcal{T}}}\Rightarrow \mathcal{N}(0,1)$.
\end{proposition}
\begin{proof}
By \eqref{eq:vhat_qT},
\[
\frac{(n_{\mathrm{eff}}-1)\widehat V}{Q_{\mathcal{T}}}
=
1-\frac{1}{n_{\mathrm{eff}}}\Big(\frac{S_{\mathcal{T}}}{\sqrt{Q_{\mathcal{T}}}}\Big)^2,
\]
which converges to $1$ in probability under the stated conditions.
\end{proof}
\paragraph{Score sum notation.}
Define the centered increments $\xi_t:=\widehat\phi_t-\theta_0$ and the score sum
$S_{\mathcal{T}}:=\sum_{t\in\mathcal{T}}\xi_t=n_{\mathrm{eff}}(\widehat\theta-\theta_0)$.

\begin{equation}\label{eq:sehat}
\widehat{\mathrm{SE}} := \sqrt{\widehat V/n_{\mathrm{eff}}}.
\end{equation}
\begin{equation}\label{eq:ci}
\mathrm{CI}_{1-\alpha} = \Big[\widehat{\theta}\pm z_{1-\alpha/2}\,\widehat{\mathrm{SE}}\Big].
\end{equation}
\begin{remark}[Finite-$n$ reporting convention]\label{rem:t_vs_z}
For small $n_{\mathrm{eff}}$, it is common as a pragmatic finite-sample convention to also report a $t$-critical variant using $t_{0.975, n_{\mathrm{eff}}-1}$, namely $CI_t:=\widehat\theta\pm t_{0.975, n_{\mathrm{eff}}-1}\sqrt{\widehat V/n_{\mathrm{eff}}}$, where $\widehat V$ is the sample variance in \eqref{eq:Vhat}. This is a reporting convention; we do not claim finite-sample Student-$t$ validity, and it is not covered by Theorem~\ref{thm:studentized}.
\end{remark}

% ============================================================================
% 4. MAIN THEORETICAL RESULTS
% ============================================================================
\section{Main Theoretical Results}\label{sec:main_results}

\begin{remark}[Triangular-array and uniformity convention]\label{rem:triangular_uniform}
All objects may depend on the horizon $n$ (e.g.\ forward blocks and nuisance
fits). We suppress the index $n$ and assume overlap and moment constants are
uniform over $n$.
\end{remark}

This section isolates the logic behind predictable AIPW/DML in adaptive experiments. First, under logged executed propensities and predictable nuisance fits, the centered score is an exact martingale difference, yielding finite-sample conditional unbiasedness (Proposition~\ref{prop:identification} and Lemma~\ref{lem:mds}). Second, fixed-horizon inference follows from a self-normalized martingale CLT applied to the martingale sum, with feasibility provided by a sample-variance approximation (Theorem~\ref{thm:studentized}). Table~\ref{tab:assumption_ledger} in Appendix~\ref{app:assumptions} maps each assumption to each statement, and Appendix~\ref{app:martingale} states the martingale limit theorem invoked in the proofs.

\subsection{The Martingale Structure}\label{sec:martingale}

The central insight of our analysis is that forward cross-fitting induces a martingale difference structure. We first establish an identification result for the pseudo-outcome.

\begin{proposition}[Identification, including predictable random nuisances]\label{prop:identification}
Let $m_{t-1}=(m_{t-1,0},m_{t-1,1})$ be any (possibly random) pair of regression functions that is $\cF_{t-1}$-measurable. Under Assumptions~\ref{ass:iid}, \ref{ass:logging_integrity}, and~\ref{ass:unconf}, for each $t$,
\begin{equation}\label{eq:phi_ident}
\E\!\big[\phi_t(m_{t-1})\mid \cG_t\big]=\tau(X_t).
\end{equation}
Consequently, $\E[\phi_t(m_{t-1})]=\theta_0$.
\end{proposition}

\begin{proof}
Fix $t$ and condition on $\cG_t=\sigma(\cF_{t-1},X_t,\pi_t)$. Since $m_{t-1}$ is
$\cF_{t-1}$-measurable and $\cF_{t-1}\subseteq \cG_t$, the function $m_{t-1}$
is fixed under this conditioning.

By sequential randomization (Assumption~\ref{ass:unconf}),
\[
\E[Y_t \mid \cG_t, A_t=1]=\E[Y_t(1)\mid \cG_t],\qquad
\E[Y_t \mid \cG_t, A_t=0]=\E[Y_t(0)\mid \cG_t].
\]
By i.i.d.\ arrivals (Assumption~\ref{ass:iid}) and the exogeneity condition in Definition~\ref{def:propensity}, $(Y_t(0),Y_t(1))\indep (\cF_{t-1},\pi_t)\mid X_t$ and thus $\E[Y_t(a)\mid X_t,\cF_{t-1},\pi_t]=\E[Y_t(a)\mid X_t]=m_a^\star(X_t)$ for $a\in\{0,1\}$.

Using $\E[A_t\mid \cG_t]=\pi_t$ and the tower property,
\begin{align*}
\E\!\left[\frac{A_t(Y_t-m_{t-1,1}(X_t))}{\pi_t}\,\Big|\,\cG_t\right]
&=
\E\!\left[\frac{A_t}{\pi_t}\,\E[Y_t-m_{t-1,1}(X_t)\mid \cG_t,A_t]\,\Big|\,\cG_t\right]\\
&=
\E\!\left[\frac{A_t}{\pi_t}\,(m_1^\star(X_t)-m_{t-1,1}(X_t))\,\Big|\,\cG_t\right]\\
&= m_1^\star(X_t)-m_{t-1,1}(X_t),
\end{align*}
and similarly
\[
\E\!\left[\frac{(1-A_t)(Y_t-m_{t-1,0}(X_t))}{1-\pi_t}\,\Big|\,\cG_t\right]
= m_0^\star(X_t)-m_{t-1,0}(X_t).
\]
Plugging into the definition of $\phi_t(m_{t-1})$ in \eqref{eq:phi_def} yields
\begin{align*}
\E[\phi_t(m_{t-1})\mid \cG_t]
&= (m_{t-1,1}(X_t)-m_{t-1,0}(X_t))
+(m_1^\star(X_t)-m_{t-1,1}(X_t))\\
&\quad -(m_0^\star(X_t)-m_{t-1,0}(X_t))\\
&= m_1^\star(X_t)-m_0^\star(X_t)
=\tau(X_t),
\end{align*}
which proves \eqref{eq:phi_ident}. Taking unconditional expectations gives
$\E[\phi_t(m_{t-1})]=\E[\tau(X_t)]=\theta_0$.
\end{proof}

The key consequence is that the pseudo-outcome, when evaluated with any predictable nuisance estimates, remains conditionally unbiased.

\begin{lemma}[Martingale difference structure]\label{lem:mds}
Suppose Assumptions~\ref{ass:iid}, \ref{ass:logging_integrity}, \ref{ass:unconf}, \ref{ass:moments}, and \ref{ass:nuis_stab} hold.
Let $\{\widehat m_{t-1}\}$ satisfy Assumption~\ref{ass:predictable_nuis}.
For each scored index $t\in\mathcal{T}$, define $\widehat\phi_t:=\phi_t(\widehat m_{t-1})$ and $\xi_t:=\widehat\phi_t-\theta_0$.
Then
\begin{equation}\label{eq:mds}
\E[\widehat\phi_t \mid \cF_{t-1}] = \theta_0,
\end{equation}
and $\{\xi_t,\cF_t\}$ is a martingale difference sequence over the scored indices.
\end{lemma}

\begin{proof}
(A finite first moment suffices here; this follows from Assumption~\ref{ass:moments}. Higher-moment bounds used later are provided by Lemma~\ref{lem:moment_bounds}.)
Using iterated expectations and Proposition~\ref{prop:identification},
\[
\E[\phi_t(\widehat m_{t-1})\mid \cF_{t-1}]
=
\E\!\big[\E[\phi_t(\widehat m_{t-1})\mid \cG_t]\mid \cF_{t-1}\big]
=
\E[\tau(X_t)\mid \cF_{t-1}].
\]
By Assumption~\ref{ass:iid}, $X_t$ is independent of $\cF_{t-1}$, so $\E[\tau(X_t) \mid \cF_{t-1}] = \E[\tau(X)] = \theta_0$.
\end{proof}

\begin{corollary}[Finite-sample unbiasedness]\label{cor:unbiased}
Under Assumptions~\ref{ass:iid}, \ref{ass:logging_integrity}, \ref{ass:unconf}, \ref{ass:T_det}, \ref{ass:moments}, \ref{ass:nuis_stab}, and predictable nuisances as in Assumption~\ref{ass:predictable_nuis}, we have $\E[\widehat{\theta}]=\theta_0$.
\end{corollary}

\begin{proof}
By linearity and Lemma~\ref{lem:mds},
\[
\E[\widehat\theta]
=
\frac{1}{n_{\mathrm{eff}}}\sum_{t\in\mathcal{T}}
\E[\phi_t(\widehat m_{t-1})]
=
\frac{1}{n_{\mathrm{eff}}}\sum_{t\in\mathcal{T}} \theta_0
=
\theta_0.
\]
\end{proof}

\begin{remark}[Interpretation]\label{rem:mds_interpretation}
Lemma~\ref{lem:mds} is the conceptual replacement for fold-wise independence in i.i.d.\ DML. Rather than requiring approximate independence between nuisance estimation and score evaluation, we exploit exact conditional unbiasedness given the past. This martingale structure holds for any quality of nuisance estimates, provided they are predictable.
\end{remark}

\begin{lemma}[Moment bounds]\label{lem:moment_bounds}
Suppose Assumptions~\ref{ass:moments}, \ref{ass:nuis_stab}, and \ref{ass:overlap} hold. Then there exists a constant $C<\infty$ such that
\[
\sup_{n\ge 1}\sup_{t\in\mathcal{T}} \E\big[|\widehat\phi_t|^4\big]\le C.
\]
\end{lemma}
\begin{proof}
Fix $t\in\mathcal{T}$. By overlap, $\pit\in[\varepsilon,1-\varepsilon]$ almost surely, so the inverse-propensity weights are uniformly bounded on scored indices. Applying $(a+b+c+d)^4\le 4^3(a^4+b^4+c^4+d^4)$ to the score formula \eqref{eq:phi_def} (and hence to \eqref{eq:phi_hat}) and using Assumptions~\ref{ass:moments} and \ref{ass:nuis_stab} yields the claimed uniform bound.
\end{proof}

\begin{remark}[Weaker than i.i.d.]\label{rem:weaker_iid}
Lemma~\ref{lem:mds} uses Assumption~\ref{ass:iid} only through the implication $\E[\tau(X_t)\mid \cF_{t-1}]=\theta_0$.
Thus the MDS argument extends to any arrival process satisfying this mean-stationarity condition.
\end{remark}

\subsection{The Variance Structure}\label{sec:variance_structure}

Before stating the CLT, we analyze the variance structure. Define the realized and predictable quadratic variations:
\begin{equation}\label{eq:quadratic_variations}
Q_{\mathcal{T}} := \sum_{t\in\mathcal{T}} \xi_t^2, \qquad
V_{\mathcal{T}}^2 := \sum_{t\in\mathcal{T}} \E[\xi_t^2 \mid \cF_{t-1}].
\end{equation}

\begin{proposition}[Conditional second-moment decomposition and oracle benchmark]\label{prop:var_decomp}
Suppose Assumptions~\ref{ass:iid}, \ref{ass:logging_integrity}, and \ref{ass:seq_rand} hold. Let $\sigma_a^2(X_t):=\Var(Y_t(a)\mid X_t)$ and define the regression bias terms $b_a(x):=m_a(x)-m_a^\star(x)$. Then for each $t\in\mathcal{T}$,
\begin{equation}\label{eq:general_second_moment}
\E[(\phi_t(m)-\theta_0)^2 \mid \cG_t]
=
(\tau(X_t)-\theta_0)^2
+\frac{\sigma_1^2(X_t)}{\pit}
+\frac{\sigma_0^2(X_t)}{1-\pit}
+\pit(1-\pit)\left(\frac{b_1(X_t)}{\pit}+\frac{b_0(X_t)}{1-\pit}\right)^2.
\end{equation}
Moreover, since $\E[\phi_t(m)\mid \cG_t]=\tau(X_t)$ for any $m$ that is $\cF_{t-1}$-measurable, the last three terms in \eqref{eq:general_second_moment} equal $\Var(\phi_t(m)\mid \cG_t)$. In particular, relative to the oracle score $\phi_t(m^\star)$ (for which $b_0=b_1\equiv 0$), using a misspecified $m$ adds a nonnegative augmentation term to the conditional second moment and conditional variance.
\end{proposition}
\begin{proof}
Define $\varepsilon_{t,a}:=Y_t(a)-m_a^\star(X_t)$, so that $\E[\varepsilon_{t,a}\mid X_t]=0$ and $\E[\varepsilon_{t,a}^2\mid X_t]=\sigma_a^2(X_t)$. A direct expansion using sequential randomization yields
\[
\phi_t(m)-\tau(X_t)
=
\frac{A_t}{\pit}\varepsilon_{t,1}-\frac{1-A_t}{1-\pit}\varepsilon_{t,0}
-\frac{A_t-\pit}{\pit}b_1(X_t)-\frac{A_t-\pit}{1-\pit}b_0(X_t),
\]
from which \eqref{eq:general_second_moment} follows by taking conditional second moments given $\cG_t$ and adding $(\tau(X_t)-\theta_0)^2$. Finally, since $\E[\phi_t(m)\mid \cG_t]=\tau(X_t)$ for any $m$ that is $\cF_{t-1}$-measurable (in particular, any predictable learner), the last three terms in \eqref{eq:general_second_moment} are all centered at the oracle benchmark $\tau(X_t)$, and the proof is complete.
\end{proof}

\begin{remark}[Oracle case and variance inflation are immediate]\label{rem:oracle_from_decomp}
Setting $m=m^\star$ in Proposition~\ref{prop:var_decomp} removes the
nonnegative augmentation term and yields the oracle conditional second moment.
The factors $1/\pi_t$ and $1/(1-\pi_t)$ make explicit how extreme propensities
inflate uncertainty, motivating overlap enforcement at the design stage.
\end{remark}

\begin{remark}[Why learn $m$ without rates?]\label{rem:why_learn_m}
Theorem~\ref{thm:studentized} yields valid prespecified-horizon coverage under predictability and logged propensities without requiring that $\widehat m_{t-1,a}$ converges to $m_a^\star$ at any particular rate. Learning $m$ is therefore not needed for validity; it is needed for precision. Proposition~\ref{prop:var_decomp} shows that the oracle score $\phi_t(m^\star)$ removes a nonnegative augmentation term from the conditional second moment, providing a natural benchmark for how nuisance quality affects the size of the studentizer. The oracle regression $m^\star$ minimizes the conditional variance within the AIPW score family (Proposition~\ref{prop:var_decomp}), and the augmentation term shrinks as $m$ approaches $m^\star$. This does not imply monotone tightening relative to an arbitrary baseline across learning iterations. There is no monotonic guarantee relative to an arbitrary baseline. A poorly chosen nuisance can increase or decrease variance compared to another misspecified choice.
\end{remark}

\subsection{Studentized Asymptotic Normality}\label{sec:studentized_clt}

We now state our main inferential result. The key ingredient is that the
realized quadratic variation $Q_{\mathcal{T}}$ and its predictable counterpart
$V_{\mathcal{T}}^2$ are asymptotically equivalent under mild conditions,
enabling studentization with the feasible sample variance $\widehat V$.

\paragraph{Roadmap.}
Lemma~\ref{lem:mds} identifies the scored sum $S_{\mathcal{T}}$ as a martingale sum with predictable quadratic variation $V_{\mathcal{T}}^2$ and realized quadratic variation $Q_{\mathcal{T}}$; Assumption~\ref{ass:T_det} ensures that the scored indicator $\ind{t\in\mathcal{T}}$ is deterministic (hence predictable). To apply the self-normalized martingale CLT in Theorem~\ref{thm:selfnorm}, we verify: (i) $V_{\mathcal{T}}^2\to\infty$ (Assumption~\ref{ass:var_growth}); (ii) $Q_{\mathcal{T}}/V_{\mathcal{T}}^2\to_p 1$ via a standard martingale variance-ratio argument, using the uniform fourth-moment bound from Lemma~\ref{lem:moment_bounds}; and (iii) a Lyapunov condition, again by Lemma~\ref{lem:moment_bounds} together with Assumption~\ref{ass:var_growth}. Finally, Proposition~\ref{prop:feasibility} justifies replacing $Q_{\mathcal{T}}$ by $(n_{\mathrm{eff}}-1)\widehat V$ in the Studentized statistic.

\begin{assumption}[Variance growth / nondegeneracy]\label{ass:var_growth}
As $n\to\infty$, $n_{\mathrm{eff}}:=|\mathcal{T}|\to\infty$ and there exists $v_->0$ such that
\[
\Pp\left(V_{\mathcal{T}}^2 \ge v_- n_{\mathrm{eff}}\right)\to 1.
\]
This is a nondegeneracy lower bound ensuring the predictable quadratic variation grows at least linearly in $n_{\mathrm{eff}}$; it does not impose variance stabilization.
\end{assumption}

A simple primitive sufficient condition for Assumption~\ref{ass:var_growth} is that outcome noise is nondegenerate and overlap holds under the standing causal/sequential assumptions of Section~\ref{sec:model}. The following corollary records a sufficient condition; see Appendix~\ref{sec:var_growth_sufficient} for the proof.

\begin{corollary}[A sufficient condition for variance growth]\label{cor:var_growth_suff}
Suppose Assumptions~\ref{ass:iid}, \ref{ass:logging_integrity}, \ref{ass:seq_rand}, \ref{ass:overlap}, and \ref{ass:noise} hold, where Assumption~\ref{ass:noise} (Appendix) is the nondegeneracy condition $\E[\sigma_1^2(X)+\sigma_0^2(X)]\ge \underline{\sigma}^2>0$. Then Assumption~\ref{ass:var_growth} holds.
\end{corollary}

\begin{remark}[Nondegeneracy and numerical guardrails]\label{rem:Vhat_pos}
Assumption~\ref{ass:var_growth} implies the quadratic variation grows, so
$\widehat V$ is bounded away from zero with high probability asymptotically.
In finite samples, if $\widehat V=0$ (e.g.\ constant outcomes), inference is
uninformative; report this as a design/measurement degeneracy.
\end{remark}

\begin{theorem}[Studentized fixed-horizon inference without variance stabilization]\label{thm:studentized}
Suppose Assumptions~\ref{ass:iid}, \ref{ass:logging_integrity}, \ref{ass:unconf}, \ref{ass:overlap}, \ref{ass:predictable_nuis}, \ref{ass:T_det}, \ref{ass:moments}, \ref{ass:nuis_stab}, and \ref{ass:var_growth} hold. Let $n_{\mathrm{eff}}:=|\mathcal{T}|$ and recall $S_{\mathcal{T}}=\sum_{t\in\mathcal{T}}(\widehat\phi_t-\theta_0)$ and $Q_{\mathcal{T}}=\sum_{t\in\mathcal{T}}(\widehat\phi_t-\theta_0)^2$. Then
\begin{equation}\label{eq:studentized_clt_q}
\frac{S_{\mathcal{T}}}{\sqrt{Q_{\mathcal{T}}}}\Rightarrow \mathcal{N}(0,1),
\end{equation}
and the practical Studentized statistic satisfies
\begin{equation}\label{eq:studentized_clt}
\frac{\sqrt{n_{\mathrm{eff}}}(\widehat\theta-\theta_0)}{\widehat V^{1/2}}\Rightarrow \mathcal{N}(0,1),
\end{equation}
where $\widehat V$ is defined in \eqref{eq:Vhat}. Consequently, the Wald interval \eqref{eq:ci} has asymptotically correct coverage at the prespecified horizon, without requiring variance stabilization (i.e., $V_{\mathcal{T}}^2/n_{\mathrm{eff}}\to V$ deterministic; Assumption~\ref{ass:stab}).
\end{theorem}
\begin{proof}
Let $\xi_t:=\ind{t\in\mathcal{T}}(\widehat\phi_t-\theta_0)$ and $S_{\mathcal{T}}:=\sum_{t\in\mathcal{T}}\xi_t$. Under Assumptions~\ref{ass:iid}, \ref{ass:logging_integrity}, \ref{ass:unconf}, \ref{ass:predictable_nuis}, and \ref{ass:T_det}, Lemma~\ref{lem:mds} implies that $(\xi_t,\mathcal{F}_t)$ is a martingale difference array. Assumption~\ref{ass:var_growth} gives $V_{\mathcal{T}}^2:=\sum_{t\in\mathcal{T}}\E[\xi_t^2\mid \mathcal{F}_{t-1}]\to\infty$. Lemma~\ref{lem:moment_bounds} (which uses Assumptions~\ref{ass:moments}, \ref{ass:nuis_stab}, and \ref{ass:overlap}) provides a uniform fourth-moment bound for the scored increments; together with Assumption~\ref{ass:var_growth}, a standard martingale variance-ratio calculation (see Appendix~\ref{app:proofs}) verifies the variance-ratio and Lyapunov conditions required by Theorem~\ref{thm:selfnorm}. Therefore Theorem~\ref{thm:selfnorm} yields $S_{\mathcal{T}}/\sqrt{Q_{\mathcal{T}}}\Rightarrow \mathcal{N}(0,1)$, i.e., \eqref{eq:studentized_clt_q}. Finally, Proposition~\ref{prop:feasibility} implies $(n_{\mathrm{eff}}-1)\widehat V/Q_{\mathcal{T}}\to_p 1$, and Slutsky's theorem yields \eqref{eq:studentized_clt}.
\end{proof}

\begin{remark}[Validity vs.\ precision and consistency]\label{rem:no_rates}
The studentized CLT in Theorem~\ref{thm:studentized} does not require
$\widehat m_{t-1}$ to converge to $m^\star$ for validity. Predictability
(Assumption~\ref{ass:predictable_nuis}) together with overlap and moment/stability conditions (Assumptions~\ref{ass:overlap}--\ref{ass:nuis_stab}) and variance growth (Assumption~\ref{ass:var_growth}) are
sufficient. Nuisance learning is nonetheless valuable: by
Proposition~\ref{prop:var_decomp} the oracle regression $m^\star$ minimizes the conditional variance within the AIPW score family (Proposition~\ref{prop:var_decomp}), and the augmentation term shrinks as $m$ approaches $m^\star$; this does not imply monotone tightening relative to an arbitrary baseline across learning iterations. In particular, $\widehat\theta$
is consistent whenever $n_{\mathrm{eff}}\to\infty$, and $\widehat V$ is a feasible
studentizer even when $V_{\mathcal{T}}^2/n_{\mathrm{eff}}$ does not stabilize.
\end{remark}

\begin{remark}[Relation to stabilization-based and anytime-valid approaches]\label{rem:stabilization}\label{rem:optional_stopping}
Some fixed-horizon CLTs in the adaptive-experiment literature normalize by a deterministic asymptotic variance, imposing stabilization/design-stability conditions ensuring that an average conditional variance (or, equivalently, $V_{\mathcal{T}}^2/n_{\mathrm{eff}}$) converges to a deterministic limit \citep{hadad2021confidence,zhan2021off,katoYasuiMcAlinn2021,cook2024semiparametric,sengupta2025designstability,zenati2025kernel}. Such assumptions are appropriate when the deployed policy stabilizes or when one explicitly engineers stability.

Theorem~\ref{thm:studentized} takes a different route. We normalize by the realized quadratic variation, so validity does not require a deterministic variance limit. This directly accommodates regimes where propensities oscillate, converge to random limits, or otherwise fail to stabilize, while still delivering a conventional fixed-horizon Wald interval. Finally, Theorem~\ref{thm:studentized} is a prespecified-horizon result. If the experiment is continuously monitored and potentially stopped early based on the data, time-uniform methods such as confidence sequences and time-uniform CLTs are needed \citep{howard2021confidence,waudbySmithRamdas2024betting,waudby2024timeuniform}. Because our score increments form an exact MDS, the setup is compatible with standard time-uniform martingale methods, but we do not develop anytime-valid inference here.
\end{remark}

Appendix~\ref{app:additional_results} collects three extensions not needed for the main methods-note claim:
(i) a non-studentized CLT under variance stabilization;
(ii) oracle equivalence under $L^2$ nuisance consistency; and
(iii) a primitive sufficient condition for variance growth.

\section{Simulation Study}\label{sec:simulation}
\normalsize

We compare fixed-horizon 95\% Wald intervals of the form \eqref{eq:ci} under different normalizations and nuisance-fitting regimes, depending on the design:
\begin{enumerate}[leftmargin=2em,label=(CI\arabic*)]
\item \textbf{SN (self-normalized):} the proposed studentized Wald interval \eqref{eq:ci}.
\item \textbf{Fixed-$V$ (stabilization-style baseline):} a Wald interval that replaces the (potentially random) long-run variance by a fixed constant $V_{\mathrm{fix}}$,
\[
\widehat\theta \pm z_{1-\alpha/2}\sqrt{V_{\mathrm{fix}}/n_{\mathrm{eff}}},
\]
thereby behaving as if Assumption~\ref{ass:stab} held with deterministic limit. We specify $V_{\mathrm{fix}}$ explicitly in each design.
\item \textbf{Leaky baselines (Designs C1 and D):} nuisance fits that violate predictability (Assumption~\ref{ass:predictable_nuis}) by using contemporaneous or future outcomes when scoring unit $t$. When computed, the score is still \eqref{eq:phi_def} and the CI is \eqref{eq:ci}, but the fixed-horizon guarantee of Theorem~\ref{thm:studentized} does not apply.
\item \textbf{Regime-aware Fixed-$V$ (Design A only):} an infeasible regime-aware fixed-variance benchmark that plugs in the correct regime-specific $V_{\mathrm{fix}}$ based on the burn-in sign.
\end{enumerate}

For each design we report: (i) two-sided 95\% coverage; (ii) average CI length; and (iii) where informative, mean bias and/or null rejection rates. Monte Carlo standard errors for coverage are \mcse{$\sqrt{\hat p(1-\hat p)/R}$}. Tables report $n$ and $n_{\mathrm{eff}}$ explicitly. Unless noted, we use $R=1000$ replications and $n\in\{250,500,1000,2000,5000\}$.

\subsection{Design A: random long-run variance}\label{sec:sim_A}

During burn-in, $\pi_t=0.5$ for $t\le n_0=50$, and we compute the burn-in estimate
\[
\widehat\tau_{\mathrm{burn}}
:=\frac{1}{n_0}\sum_{t=1}^{n_0}\left(\frac{A_tY_t}{\pi_t}-\frac{(1-A_t)Y_t}{1-\pi_t}\right)
\]
using only burn-in units. After burn-in, if $\widehat\tau_{\mathrm{burn}}\ge 0$ then $\pi_t=0.8$, else $\pi_t=0.2$. Outcomes: $Y(0)\sim\mathcal{N}(0,1)$ and $Y(1)=\varepsilon_1$ with $\varepsilon_1\sim\mathcal{N}(0,9)$. There are no covariates; with $m_0\equiv m_1\equiv 0$, the AIPW score \eqref{eq:phi_def} reduces to the usual inverse-propensity weighted (IPW) score.

On $\mathcal{T}=\{n_0+1,\dots,n\}$ we have $n_{\mathrm{eff}}=n-n_0$ and $\pi_t$ is constant after burn-in, so the oracle conditional variance converges to a random limit:
\[
\frac{V_{\mathcal{T}}^2}{n_{\mathrm{eff}}}\ \to\
\begin{cases}
9/0.8 + 1/0.2 = 16.25, & \widehat\tau_{\mathrm{burn}}\ge 0,\\
9/0.2 + 1/0.8 = 46.25, & \widehat\tau_{\mathrm{burn}}<0.
\end{cases}
\]
Thus Assumption~\ref{ass:stab} fails (no deterministic variance limit). The SN CI remains valid under Theorem~\ref{thm:studentized} because it self-normalizes by realized quadratic variation. SN uses \eqref{eq:ci}. Fixed-$V$ sets $V_{\mathrm{fix}}:=31.25$ (the unconditional mean of the two variance limits above). Regime-aware Fixed-$V$ plugs in $V_{\mathrm{fix}}=16.25$ or $46.25$ depending on the realized burn-in sign.

\begin{table}[!htbp]
\centering
\small
\begin{threeparttable}
\caption{Design A: coverage and length of 95\% CIs.}
\label{tab:simA_main}
\begin{tabular}{@{}l r l r r r@{}}
\toprule
{$n$} & {$n_{\mathrm{eff}}$} & {Method} & {Coverage} & {MCSE} & {Avg.\ length} \\
\midrule
250 & 200 & Fixed-V & 0.944 & \mcse{0.007} & 1.549 \\
250 & 200 & Regime-Fixed & 0.950 & \mcse{0.007} & 1.515 \\
250 & 200 & SN & 0.955 & \mcse{0.007} & 1.506 \\
500 & 450 & Fixed-V & 0.939 & \mcse{0.008} & 1.033 \\
500 & 450 & Regime-Fixed & 0.950 & \mcse{0.007} & 1.003 \\
500 & 450 & SN & 0.957 & \mcse{0.006} & 1.002 \\
1000 & 950 & Fixed-V & 0.938 & \mcse{0.008} & 0.711 \\
1000 & 950 & Regime-Fixed & 0.948 & \mcse{0.007} & 0.697 \\
1000 & 950 & SN & 0.948 & \mcse{0.007} & 0.696 \\
2000 & 1950 & Fixed-V & 0.947 & \mcse{0.007} & 0.496 \\
2000 & 1950 & Regime-Fixed & 0.941 & \mcse{0.007} & 0.477 \\
2000 & 1950 & SN & 0.946 & \mcse{0.007} & 0.476 \\
5000 & 4950 & Fixed-V & 0.941 & \mcse{0.007} & 0.311 \\
5000 & 4950 & Regime-Fixed & 0.960 & \mcse{0.006} & 0.303 \\
5000 & 4950 & SN & 0.959 & \mcse{0.006} & 0.303 \\
\bottomrule
\end{tabular}
\begin{tablenotes}[flushleft]
\footnotesize
\item Notes: $R=1000$ replications; $n\in\{250,500,1000,2000,5000\}$; burn-in $n_0=50$ (so $n_{\mathrm{eff}}=n-50$). SN is the studentized interval \eqref{eq:ci}. Fixed-$V$ uses $V_{\mathrm{fix}}=31.25$, the unconditional mean of the regime-specific variance limits. Regime-aware Fixed-$V$ is an infeasible oracle benchmark that plugs in the correct regime-specific variance term (i.e., $\sigma_1^2=9$ and $\sigma_0^2=1$, so the regime-specific limits are $9/0.8+1/0.2=16.25$ and $9/0.2+1/0.8=46.25$), and uses $z$-critical values.
\end{tablenotes}
\end{threeparttable}
\end{table}

Table~\ref{tab:simA_main} reports marginal coverage and average length for Design~A as $n$ increases. In this design, the studentized SN interval remains close to the nominal $95\%$ level across horizons, consistent with Theorem~\ref{thm:studentized}. The fixed-$V$ interval normalizes by a single deterministic target $V$ that averages over the two post--burn-in propensity regimes. As a consequence, its marginal coverage can be close to nominal while still exhibiting regime-dependent over/under coverage (Table~\ref{tab:simA_conditional}). The regime-aware fixed-$V$ benchmark, which plugs in the correct regime-specific variance constant, illustrates that classical normalization works once one conditions on (and correctly accounts for) the realized variance regime. The SN interval achieves this conditioning automatically via studentization.

\begin{table}[t]
\centering
\small
\begin{threeparttable}
\begin{tabular}{lccccc}
\toprule
Method & Post burn-in $\pi$ & $n=250$ & $n=500$ & $n=1000$ & $n=2000$ \\
\midrule
SN & $0.8$ & 0.96 & 0.97 & 0.95 & 0.95 \\
SN & $0.2$ & 0.95 & 0.95 & 0.94 & 0.94 \\
Fixed-$V$ & $0.8$ & 1.00 & 0.99 & 0.99 & 1.00 \\
Fixed-$V$ & $0.2$ & 0.89 & 0.89 & 0.89 & 0.89 \\
\bottomrule
\end{tabular}
\begin{tablenotes}
\small
\item Notes: Each row conditions on the realized post--burn-in propensity in Design~A: if $\widehat\tau_{\mathrm{burn}}\ge 0$ then $\pi_t=0.8$ for $(t>n_0)$, and if $\widehat\tau_{\mathrm{burn}}<0$ then $\pi_t=0.2$ for $(t>n_0)$.
\end{tablenotes}
\end{threeparttable}
\caption{Conditional coverage by realized propensity regime (Design~A)}\label{tab:simA_conditional}
\end{table}

Design~A illustrates the basic phenomenon targeted by Theorem~\ref{thm:studentized}: the variance proxy $V_{\mathcal{T}}^2/n_{\mathrm{eff}}$ does not converge to a single deterministic constant because it depends on the realized post--burn-in propensity regime. Table~\ref{tab:simA_conditional} shows that a fixed-$V$ normalization can be conservative when the realized regime has low variance ($\pi_t=0.8$) and anti-conservative when the realized regime has high variance ($\pi_t=0.2$). In contrast, the SN interval studentizes by the realized quadratic variation and maintains near-nominal conditional coverage in both regimes.

\subsection{Design B: stabilization benchmark}\label{sec:sim_B}

Design~B provides a benchmark setting in which the conditional variance stabilizes. We adopt the outcome model from Design~A with $\tau=0$, $Y(0)\sim\mathcal{N}(0,1)$, and $Y(1)=\varepsilon_1$ where $\varepsilon_1\sim\mathcal{N}(0,9)$. Treatment is assigned with a constant executed propensity $\pi_t\equiv 0.6$ for all $t\le n$, so there is no burn-in and $n_{\mathrm{eff}}=n$. Because there are no covariates, the AIPW score coincides with the IPW score. In this setting, Assumption~\ref{ass:stab} holds with a deterministic long-run variance,
\[
V_{\mathrm{fix}} \;=\; \frac{9}{0.6} + \frac{1}{0.4} \;=\; 17.5,
\]
and therefore the SN and Fixed-$V$ intervals are asymptotically equivalent.

\begin{table}[!htbp]
\centering
\small
\begin{threeparttable}
\caption{Design B (stabilization benchmark): coverage and length of 95\% CIs.}
\label{tab:simB}
\begin{tabular}{@{}l r l r r r@{}}
\toprule
{$n$} & {$n_{\mathrm{eff}}$} & {Method} & {Coverage} & {MCSE} & {Avg.\ length} \\
\midrule
250 & 250 & Fixed-V & 0.954 & \mcse{0.007} & 1.037 \\
250 & 250 & SN & 0.952 & \mcse{0.007} & 1.038 \\
500 & 500 & Fixed-V & 0.949 & \mcse{0.007} & 0.733 \\
500 & 500 & SN & 0.953 & \mcse{0.007} & 0.733 \\
1000 & 1000 & Fixed-V & 0.957 & \mcse{0.006} & 0.519 \\
1000 & 1000 & SN & 0.955 & \mcse{0.007} & 0.518 \\
2000 & 2000 & Fixed-V & 0.953 & \mcse{0.007} & 0.367 \\
2000 & 2000 & SN & 0.953 & \mcse{0.007} & 0.366 \\
5000 & 5000 & Fixed-V & 0.947 & \mcse{0.007} & 0.232 \\
5000 & 5000 & SN & 0.947 & \mcse{0.007} & 0.232 \\
\bottomrule
\end{tabular}
\begin{tablenotes}[flushleft]
\footnotesize
\item Notes: $R=1000$ replications; $n\in\{250,500,1000,2000,5000\}$; $n_{\mathrm{eff}}=n$. SN is \eqref{eq:ci}. Fixed-$V$ uses $V_{\mathrm{fix}}=17.5$.
\end{tablenotes}
\end{threeparttable}
\end{table}

Table~\ref{tab:simB} shows that when the variance stabilizes, SN and Fixed-$V$ give nearly identical coverage and length. Differences are within Monte Carlo error, consistent with the deterministic $V_{\mathrm{fix}}$.

\subsection{Design C1: leakage stress test}\label{sec:sim_C1}

Design~C1 is a stress test for Assumption~\ref{ass:predictable_nuis} and Remark~\ref{rem:predictability_cf}. Covariates are $X_t\in\mathbb{R}^{p}$ with $p=20$ and i.i.d.\ $X_t\sim\mathcal{N}(0,I_p)$. Potential outcomes follow
\[
Y_t(0)=m_0(X_t)+\varepsilon_{t0},\qquad
Y_t(1)=m_0(X_t)+\tau(X_t)+\varepsilon_{t1},
\]
where $m_0(x)$ is nonlinear and $\tau(x)=\tau_0+\delta\sin(x_1)$ so the superpopulation ATE equals $\theta_0=\E[\tau(X_t)]=\tau_0$. We report calibration at $\tau_0=0$ with $\delta=0.2$; errors are independent $t_{30}$ draws rescaled to unit variance (finite moments, but heavier tails than Gaussian).

To induce adaptive feedback, we split $\{1,\ldots,n\}$ into $K=5$ contiguous blocks and use the first block as burn-in with $\pi_t\equiv 0.5$. On subsequent blocks, the executed propensity is
\[
\pi_t=\varepsilon+(1-2\varepsilon)\expit\big(\lambda\,\widehat\tau_{t-1}(X_t)\big),\qquad
\widehat\tau_{t-1}(x)=\widehat m_{t-1,1}(x)-\widehat m_{t-1,0}(x),
\]
with $\varepsilon=0.1$ and $\lambda=2.5$. The predictable nuisance $\widehat m_{t-1,a}$ is fit by ridge regression on a low-dimensional feature map using only prior blocks and held fixed within the current block (forward fitting). The leaky baseline fits a richer ridge model once on the full sample, violating predictability.

\begin{table}[!htbp]
\centering
\small
\begin{threeparttable}
\caption{Design C1 (predictability / leakage stress test, $\tau_0=0$): coverage, length, and null rejection.}
\label{tab:simC1}
\begin{tabular}{@{}l r l r r r r@{}}
\toprule
{$n$} & {$n_{\mathrm{eff}}$} & {Method} & {Coverage} & {MCSE} & {Avg.\ length} & {Reject rate} \\
\midrule
250 & 200 & SN-AIPW-Predictable & 0.952 & \mcse{0.007} & 0.653 & 0.048 \\
250 & 200 & SN-AIPW-LeakyFull & 0.821 & \mcse{0.012} & 1.864 & 0.179 \\
250 & 200 & SN-IPW & 0.939 & \mcse{0.008} & 0.766 & 0.061 \\
500 & 400 & SN-AIPW-Predictable & 0.953 & \mcse{0.007} & 0.429 & 0.047 \\
500 & 400 & SN-AIPW-LeakyFull & 0.887 & \mcse{0.010} & 0.398 & 0.113 \\
500 & 400 & SN-IPW & 0.958 & \mcse{0.006} & 0.525 & 0.042 \\
1000 & 800 & SN-AIPW-Predictable & 0.948 & \mcse{0.007} & 0.290 & 0.052 \\
1000 & 800 & SN-AIPW-LeakyFull & 0.928 & \mcse{0.008} & 0.274 & 0.072 \\
1000 & 800 & SN-IPW & 0.947 & \mcse{0.007} & 0.360 & 0.053 \\
2000 & 1600 & SN-AIPW-Predictable & 0.957 & \mcse{0.006} & 0.202 & 0.043 \\
2000 & 1600 & SN-AIPW-LeakyFull & 0.944 & \mcse{0.007} & 0.195 & 0.056 \\
2000 & 1600 & SN-IPW & 0.962 & \mcse{0.006} & 0.251 & 0.038 \\
5000 & 4000 & SN-AIPW-Predictable & 0.944 & \mcse{0.007} & 0.126 & 0.056 \\
5000 & 4000 & SN-AIPW-LeakyFull & 0.941 & \mcse{0.007} & 0.124 & 0.059 \\
5000 & 4000 & SN-IPW & 0.951 & \mcse{0.007} & 0.157 & 0.049 \\
\bottomrule
\end{tabular}
\begin{tablenotes}[flushleft]
\footnotesize
\item Notes: $R=1000$ replications; $K=5$ contiguous blocks with the first block omitted from scoring (so $n_{\mathrm{eff}}=4n/5$). The executed propensity is $\pi_t=\varepsilon+(1-2\varepsilon)\expit(\lambda\,\widehat\tau_{t-1}(X_t))$ with $\varepsilon=0.1$ and $\lambda=2.5$. ``SN-AIPW-Predictable'' uses forward (past-only) ridge nuisances; ``SN-AIPW-LeakyFull'' fits a full-sample ridge nuisance on richer features and is not predictable. Reject rate is the two-sided rejection probability of $H_0:\theta_0=0$ based on the corresponding 95\% CI.
\end{tablenotes}
\end{threeparttable}
\end{table}

Table~\ref{tab:simC1} makes the predictability requirement operational: the predictable SN-AIPW procedure remains near nominal, while the leaky full-sample nuisance substantially undercovers and over-rejects at smaller horizons, consistent with the warning in Remark~\ref{rem:predictability_cf}.

\subsection{Design C2: nuisance quality}\label{sec:sim_C2}

Design~C2 isolates how nuisance quality affects precision without confounding from adaptivity. Covariates $X_t\in\mathbb{R}^5$ follow a correlated normal with AR(1) correlation $\rho=0.5$; treatment is assigned with constant propensity $\pi_t\equiv 0.5$. Potential outcomes are linear with homoskedastic noise:
\[
Y_t(0)=X_{t1}+X_{t2}+\varepsilon_{t0},\qquad
Y_t(1)=X_{t1}+X_{t2}+\tau+\varepsilon_{t1},
\qquad \varepsilon_{t0},\varepsilon_{t1}\stackrel{\text{i.i.d.}}{\sim}\mathcal{N}(0,1),
\]
so the oracle regression removes the augmentation term in Proposition~\ref{prop:var_decomp}. We use $K=10$ forward blocks (first block omitted) and report results at $n=5000$ under $\tau=0$.

\begin{table}[!htbp]
\centering
\small
\begin{threeparttable}
\caption{Design C2 (precision vs nuisance quality, $n=5000$): coverage and relative length/variance vs oracle.}
\label{tab:simC2}
\begin{tabular}{@{}l r l r r r r r@{}}
\toprule
{$n$} & {$n_{\mathrm{eff}}$} & {Method} & {Coverage} & {MCSE} & {Avg.\ length} & {Rel.\ length} & {Rel.\ $\widehat V$} \\
\midrule
5000 & 4500 & SN-AIPW-Oracle & 0.955 & \mcse{0.007} & 0.117 & 1.000 & 1.000 \\
5000 & 4500 & SN-AIPW-WellSpec & 0.951 & \mcse{0.007} & 0.117 & 1.002 & 1.004 \\
5000 & 4500 & SN-AIPW-Misspec & 0.956 & \mcse{0.006} & 0.155 & 1.323 & 1.750 \\
5000 & 4500 & SN-IPW & 0.937 & \mcse{0.008} & 0.234 & 1.999 & 3.997 \\
\bottomrule
\end{tabular}
\begin{tablenotes}[flushleft]
\footnotesize
\item Notes: $R=1000$ replications; $\pi_t\equiv 0.5$; $K=10$ forward blocks with the first block omitted from scoring (so $n_{\mathrm{eff}}=4500$). ``WellSpec'' fits arm-specific linear regression on all covariates; ``Misspec'' fits on a restricted covariate set; ``SN-IPW'' sets $m\equiv 0$. Relative quantities are computed with respect to the oracle procedure within the same design.
\end{tablenotes}
\end{threeparttable}
\end{table}

Table~\ref{tab:simC2} shows the oracle benchmarking message: coverage remains near nominal across nuisance choices, while interval length (and $\widehat V$) inflates as the nuisance becomes more misspecified, consistent with Proposition~\ref{prop:var_decomp} and the ``validity vs precision'' discussion following Theorem~\ref{thm:studentized}.

\subsection{Design D: adaptive policy and logging integrity}\label{sec:sim_D}

Design~D is a simple contextual-bandit-style adaptive experiment that highlights the design-time role of executed propensity logging (Assumption~\ref{ass:logging_integrity}). Covariates are $X_t\in\mathbb{R}^{10}$ with correlated normal $X_t\sim\mathcal{N}(0,\Sigma)$ where $\Sigma_{ij}=\rho^{|i-j|}$ and $\rho=0.3$. The baseline outcome, heterogeneous treatment effect, and heteroskedastic noise are
\[
m_0(x)=0.8x_1+0.5x_2^2-0.5\cos(x_3)+0.25x_4,\qquad
\tau(x)=0.5x_1+0.5\sin(x_2)+0.25\ind{x_3>0}-0.25x_4x_5,
\]
\[
\sigma(x)=1+0.5|x_1|,\qquad
Y_t(a)=m_a(X_t)+\varepsilon_{ta},\ \ \varepsilon_{ta}\mid X_t\sim \mathcal{N}(0,\sigma(X_t)^2),
\]
so $\theta_0=\E[\tau(X_t)]$ is the ATE. We use burn-in $n_0=100$ with $\pi_t\equiv 0.5$ and then update in blocks of size 100: on each block we fit arm-specific linear regressions on past data only and choose the executed propensity by either an $\varepsilon$-greedy rule or a softmax rule, clipping to $[0.05,0.95]$ to enforce overlap.

We compare: (i) SN-AIPW with predictable (past-only) nuisance fits; (ii) Naive-iid-DML using leaky 5-fold cross-fitting that ignores time ordering; (iii) SN-IPW with $m\equiv 0$; (iv) SN-Oracle using the true $m_0,m_1$; and (v) SN-IPW-Assume0p5, an analysis-time mis-logging baseline that computes the score using $\pi_t\equiv 0.5$ rather than the logged executed propensity $\pi_t$.

\begin{table}[!htbp]
\centering
\small
\begin{threeparttable}
\caption{Design D (adaptive assignment, $\varepsilon$-greedy policy): coverage, length, and bias of 95\% CIs.}
\label{tab:simD_eps}
\begin{tabular}{@{}l r l r r r r@{}}
\toprule
{$n$} & {$n_{\mathrm{eff}}$} & {Method} & {Coverage} & {MCSE} & {Avg.\ length} & {Bias} \\
\midrule
250 & 150 & SN-Oracle & 0.946 & \mcse{0.007} & 1.516 & 0.009 \\
250 & 150 & SN-AIPW & 0.944 & \mcse{0.007} & 1.906 & 0.019 \\
250 & 150 & Naive-iid-DML & 0.944 & \mcse{0.007} & 1.883 & 0.019 \\
250 & 150 & SN-IPW & 0.934 & \mcse{0.008} & 2.083 & 0.003 \\
250 & 150 & SN-IPW-Assume0p5 & 0.433 & \mcse{0.016} & 1.296 & 0.703 \\
500 & 400 & SN-Oracle & 0.950 & \mcse{0.007} & 0.940 & -0.008 \\
500 & 400 & SN-AIPW & 0.939 & \mcse{0.008} & 1.174 & -0.012 \\
500 & 400 & Naive-iid-DML & 0.943 & \mcse{0.007} & 1.166 & -0.014 \\
500 & 400 & SN-IPW & 0.951 & \mcse{0.007} & 1.304 & -0.017 \\
500 & 400 & SN-IPW-Assume0p5 & 0.064 & \mcse{0.008} & 0.795 & 0.869 \\
1000 & 900 & SN-Oracle & 0.948 & \mcse{0.007} & 0.629 & 0.003 \\
1000 & 900 & SN-AIPW & 0.954 & \mcse{0.007} & 0.784 & -0.006 \\
1000 & 900 & Naive-iid-DML & 0.950 & \mcse{0.007} & 0.786 & -0.006 \\
1000 & 900 & SN-IPW & 0.947 & \mcse{0.007} & 0.864 & -0.003 \\
1000 & 900 & SN-IPW-Assume0p5 & 0.005 & \mcse{0.002} & 0.528 & 0.943 \\
2000 & 1900 & SN-Oracle & 0.953 & \mcse{0.007} & 0.433 & 0.001 \\
2000 & 1900 & SN-AIPW & 0.955 & \mcse{0.007} & 0.543 & 0.001 \\
2000 & 1900 & Naive-iid-DML & 0.948 & \mcse{0.007} & 0.548 & 0.001 \\
2000 & 1900 & SN-IPW & 0.952 & \mcse{0.007} & 0.596 & -0.003 \\
2000 & 1900 & SN-IPW-Assume0p5 & 0.001 & \mcse{0.001} & 0.363 & 0.979 \\
5000 & 4900 & SN-Oracle & 0.958 & \mcse{0.006} & 0.271 & -0.001 \\
5000 & 4900 & SN-AIPW & 0.955 & \mcse{0.007} & 0.342 & 0.000 \\
5000 & 4900 & Naive-iid-DML & 0.954 & \mcse{0.007} & 0.347 & 0.001 \\
5000 & 4900 & SN-IPW & 0.955 & \mcse{0.007} & 0.372 & -0.002 \\
5000 & 4900 & SN-IPW-Assume0p5 & 0.003 & \mcse{0.002} & 0.226 & 0.984 \\
\bottomrule
\end{tabular}
\begin{tablenotes}[flushleft]
\footnotesize
\item Notes: $R=1000$ replications; burn-in $n_0=100$ so $n_{\mathrm{eff}}=n-100$; blocks of size 100. The executed propensity is chosen by an $\varepsilon$-greedy rule with $\varepsilon=0.1$ and then clipped to $[0.05,0.95]$. Bias is $\E[\widehat\theta-\theta_0]$ (Monte Carlo mean), where $\theta_0=\E[\tau(X)]$ is approximated by Monte Carlo integration under the covariate distribution.
\end{tablenotes}
\end{threeparttable}
\end{table}

\begin{table}[!htbp]
\centering
\small
\begin{threeparttable}
\caption{Design D (adaptive assignment, softmax policy): coverage, length, and bias of 95\% CIs.}
\label{tab:simD_softmax}
\begin{tabular}{@{}l r l r r r r@{}}
\toprule
{$n$} & {$n_{\mathrm{eff}}$} & {Method} & {Coverage} & {MCSE} & {Avg.\ length} & {Bias} \\
\midrule
250 & 150 & SN-Oracle & 0.953 & \mcse{0.007} & 1.509 & -0.003 \\
250 & 150 & SN-AIPW & 0.948 & \mcse{0.007} & 1.953 & 0.012 \\
250 & 150 & Naive-iid-DML & 0.940 & \mcse{0.008} & 1.934 & 0.013 \\
250 & 150 & SN-IPW & 0.916 & \mcse{0.009} & 2.115 & 0.002 \\
250 & 150 & SN-IPW-Assume0p5 & 0.395 & \mcse{0.015} & 1.298 & 0.731 \\
500 & 400 & SN-Oracle & 0.960 & \mcse{0.006} & 0.937 & 0.004 \\
500 & 400 & SN-AIPW & 0.958 & \mcse{0.006} & 1.255 & 0.002 \\
500 & 400 & Naive-iid-DML & 0.959 & \mcse{0.006} & 1.250 & 0.006 \\
500 & 400 & SN-IPW & 0.947 & \mcse{0.007} & 1.338 & 0.004 \\
500 & 400 & SN-IPW-Assume0p5 & 0.088 & \mcse{0.009} & 0.795 & 0.840 \\
1000 & 900 & SN-Oracle & 0.952 & \mcse{0.007} & 0.637 & -0.005 \\
1000 & 900 & SN-AIPW & 0.958 & \mcse{0.006} & 0.881 & -0.011 \\
1000 & 900 & Naive-iid-DML & 0.950 & \mcse{0.007} & 0.896 & -0.008 \\
1000 & 900 & SN-IPW & 0.947 & \mcse{0.007} & 0.919 & 0.003 \\
1000 & 900 & SN-IPW-Assume0p5 & 0.009 & \mcse{0.003} & 0.530 & 0.930 \\
2000 & 1900 & SN-Oracle & 0.957 & \mcse{0.006} & 0.452 & 0.007 \\
2000 & 1900 & SN-AIPW & 0.951 & \mcse{0.007} & 0.636 & 0.001 \\
2000 & 1900 & Naive-iid-DML & 0.954 & \mcse{0.007} & 0.654 & 0.001 \\
2000 & 1900 & SN-IPW & 0.935 & \mcse{0.008} & 0.651 & 0.006 \\
2000 & 1900 & SN-IPW-Assume0p5 & 0.004 & \mcse{0.002} & 0.365 & 0.982 \\
5000 & 4900 & SN-Oracle & 0.954 & \mcse{0.007} & 0.290 & 0.001 \\
5000 & 4900 & SN-AIPW & 0.954 & \mcse{0.007} & 0.421 & 0.005 \\
5000 & 4900 & Naive-iid-DML & 0.952 & \mcse{0.007} & 0.435 & 0.005 \\
5000 & 4900 & SN-IPW & 0.944 & \mcse{0.007} & 0.417 & 0.009 \\
5000 & 4900 & SN-IPW-Assume0p5 & 0.001 & \mcse{0.001} & 0.227 & 1.015 \\
\bottomrule
\end{tabular}
\begin{tablenotes}[flushleft]
\footnotesize
\item Notes: Same design as Table~\ref{tab:simD_eps}, but the executed propensity uses a softmax rule $\pi_t=\expit(\widehat\tau_{t-1}(X_t)/T)$ with temperature $T=0.5$ and clipping to $[0.05,0.95]$.
\end{tablenotes}
\end{threeparttable}
\end{table}

Tables~\ref{tab:simD_eps}--\ref{tab:simD_softmax} show two complementary messages: with the logged executed propensities and predictable nuisance fitting, the studentized intervals are close to nominal under both adaptive policies; in contrast, the mis-logged ``Assume0p5'' baseline is severely biased and has catastrophic undercoverage, illustrating why Assumption~\ref{ass:logging_integrity} is design-critical.

\subsection{Variance-ratio variability}\label{sec:sim_fig}
\begin{figure}[t]
\centering
\IfFileExists{figures/var_ratio_hist.pdf}{\includegraphics[width=0.75\textwidth]{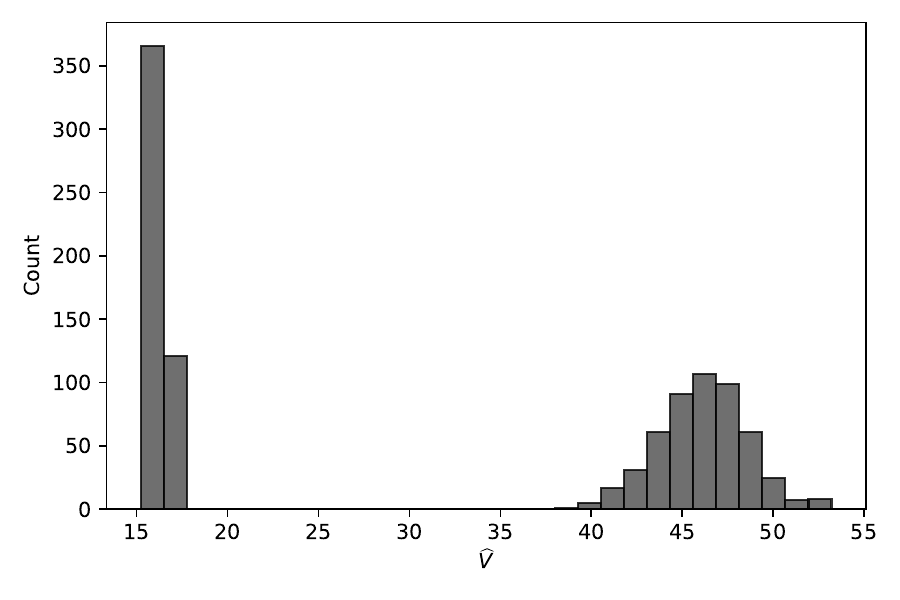}}{\fbox{Missing figure: \texttt{\detokenize{figures/var_ratio_hist.pdf}}}}
\caption{Design A: distribution across replications of the regime-dependent long-run variance proxy. In this design, $V_{\mathcal{T}}^2/n_{\mathrm{eff}}$ converges to $16.25$ or $46.25$ depending on the burn-in sign, so there is no deterministic variance limit (Assumption~\ref{ass:stab} fails). This motivates studentization by realized quadratic variation in the SN CI \eqref{eq:ci}.}
\label{fig:var_ratio_hist}
\end{figure}

Figure~\ref{fig:var_ratio_hist} shows two distinct variance modes for $V_{\mathcal{T}}^2/n_{\mathrm{eff}}$ induced by the post--burn-in propensity regime in Design~A. This bimodality explains the regime-dependent length and coverage patterns in Table~\ref{tab:simA_conditional}: intervals that plug in a single deterministic variance target are too short in the high-variance regime and too long in the low-variance regime, even if their marginal coverage can be closer to nominal by averaging over regimes. The SN interval uses the realized quadratic variation (equivalently $\widehat V$) within each replication and therefore adapts to the realized variance regime by construction.

\section{Conclusion}\label{sec:conclusion}

Modern adaptive experiments often retain a classical reporting requirement: a single end-of-study confidence interval for the superpopulation ATE $\theta_0$ at a prespecified horizon. The difficulty is not estimation per se, but calibration. When assignment probabilities evolve, the predictable quadratic variation of AIPW/DML score increments can remain replication-random, so a Wald statistic normalized by a deterministic variance target can be well behaved marginally yet miscalibrated conditional on the realized propensity regime. Our analysis keeps the usual logged-propensity AIPW/DML estimator fixed and instead enforces a martingale representation through an auditable scoring contract. The contract is minimal but concrete: the experiment log must record, unit by unit, the executed propensity actually passed to the randomization device (Assumption~\ref{ass:logging_integrity}), and the nuisance regressions used to score unit~$t$ must be measurable with respect to the past history (Assumption~\ref{ass:predictable_nuis}). Forward cross-fitting provides an implementation pattern that satisfies this predictability requirement while allowing flexible learning.

Under the standard causal model, i.i.d.\ arrivals, and executed overlap on a prespecified scored set, these pipeline conditions imply that centered AIPW/DML score increments form an exact martingale difference sequence (Lemma~\ref{lem:mds}). This step replaces the fold-wise independence logic common in i.i.d.\ DML with exact conditional unbiasedness given the past. With the martingale structure in place, inference follows by studentization along the realized experiment path. Theorem~\ref{thm:studentized} applies self-normalized martingale limit theory to show that the Studentized score sum, normalized by realized quadratic variation, converges to $\mathcal{N}(0,1)$ at the prespecified horizon even when the variance does not stabilize to any deterministic limit. Proposition~\ref{prop:feasibility} justifies the practical studentizer: the sample-variance plug-in used in standard reporting is asymptotically equivalent to realized quadratic variation under the stated conditions.

Nuisance learning enters only through precision. Proposition~\ref{prop:var_decomp} decomposes the conditional second moment, yielding an oracle benchmark and isolating a nonnegative augmentation term attributable to regression error; learning better outcome regressions reduces this term and shrinks the realized quadratic variation. Consistency is not needed for fixed-horizon validity, but under $L^2$ convergence the feasible procedure is asymptotically equivalent to the oracle AIPW/DML score (Theorem~\ref{thm:oracle_equivalence}).

The simulation designs in Section~\ref{sec:simulation} make the operational messages tangible: deterministic-variance normalizations can under- and over-cover across realized variance regimes, predictable scoring avoids leakage-induced failures, and analysis based on mis-logged propensities can be severely biased. These examples motivate the logging-and-fitting protocol and formalized as a data-contract checklist in Appendix~\ref{app:repro}. In practice, this means persisting the executed propensity and randomization metadata for each unit, enforcing overlap at assignment time on the scored set, recording the exact training indices and randomness used for each nuisance fit, and producing intervals only at the prespecified horizon so that ``peeking'' and optional stopping are excluded by design.

The contract is auditable, but it is not a substitute for causal assumptions. Logged propensities and predictable fitting can be checked from stored artifacts, whereas no interference, sequential randomization, and the no-selection/i.i.d.\ arrival condition (or its mean-stationarity alternative) are substantive and not fully testable from logs (Remark~\ref{rem:weaker_iid}). Likewise, our conclusions are strictly fixed-horizon: the intervals are not anytime-valid and should not be used under continuous monitoring or data-dependent stopping. When time-uniform guarantees are required, confidence sequences and related martingale methods are the appropriate alternatives.

Several extensions are immediate within the present framework. One is to develop time-uniform analogues that leverage the same martingale score construction while changing only the inferential target. A second is to formalize predictable scored-set rules (Remark~\ref{rem:T_predictable}) and outcome-delay settings within the same audit contract. Further work could relax the fourth-moment requirement to the $(2+\delta)$-moment regime highlighted in Remark~\ref{rem:moment_conditions}, and broaden the arrival model beyond i.i.d.\ using the mean-stationarity condition noted in Remark~\ref{rem:weaker_iid}.

\appendix
\normalsize

\newpage

\section{Martingale Limit Theory}\label{app:martingale}\label{app:selfnorm_clt}
We use a self-normalized martingale CLT for (possibly triangular-array) martingale differences. Background and related results on martingale CLTs and self-normalization can be found in \citet{hallHeyde1980}, the monograph \citet{delapena2009self}, and the survey \citet{shao2013survey}; see also modern bounds and concentration results for self-normalized martingales in \citet{fanShao2017,fanGramaLiuShao2019,bercuTouati2019}. The statement below is a direct corollary of \citet[Theorem~3.10]{hallHeyde1980} combined with Slutsky; we record it in the form used in our proof of Theorem~\ref{thm:studentized}.

\begin{theorem}[Self-normalized martingale CLT (Lyapunov form)]\label{thm:selfnorm}
Let $(\xi_{n,t},\cF_{n,t})$ be a martingale difference array with partial sums $S_n=\sum_{t=1}^{k_n}\xi_{n,t}$. Define the predictable and realized quadratic variations
\[
V_n^2:=\sum_{t=1}^{k_n}\E[\xi_{n,t}^2\mid \cF_{n,t-1}]
\qquad\text{and}\qquad
Q_n:=\sum_{t=1}^{k_n}\xi_{n,t}^2.
\]
Assume: (i) $V_n^2\to_p \infty$; (ii) $\Pp(Q_n>0)\to 1$ and $Q_n/V_n^2\to_p 1$; and (iii) for some $\delta>0$,
\[
\frac{1}{V_n^{2+\delta}}\sum_{t=1}^{k_n}\E[|\xi_{n,t}|^{2+\delta}]\to 0.
\]
Then $S_n/\sqrt{Q_n}\Rightarrow \mathcal{N}(0,1)$.
\end{theorem}

\section{Proofs}\label{app:proofs}

\subsection{Proof of Theorem~\ref{thm:studentized}}\label{app:proof_studentized}
\begin{proof}[Proof of Theorem~\ref{thm:studentized}]
The proof proceeds in two steps.

\medskip\noindent\emph{Indexing convention.}
To match Theorem~\ref{thm:selfnorm} (stated for sums over $t=1,\dots,n$), define the extended increments
\[
\tilde\xi_t := T_t(\widehat\phi_t-\theta_0), \qquad T_t:=\ind{t\in\mathcal{T}},
\qquad t=1,\dots,n.
\]
Since $\mathcal{T}$ is deterministic, $T_t$ is nonrandom and
$\E[\tilde\xi_t\mid\mathcal{F}_{t-1}]=T_t\E[\widehat\phi_t-\theta_0\mid\mathcal{F}_{t-1}]=0$,
so $\{\tilde\xi_t,\mathcal{F}_t\}$ is a martingale difference array.
$\sum_{t=1}^n \tilde\xi_t = S_{\mathcal{T}}$, $\sum_{t=1}^n \tilde\xi_t^{\,2} = Q_{\mathcal{T}}$,
and $\sum_{t=1}^n \E[\tilde\xi_t^{\,2}\mid \mathcal{F}_{t-1}] = V_{\mathcal{T}}^2$.
For readability we drop the tilde notation below.

\medskip\noindent\textbf{Step 1: Self-normalized CLT for $S_{\mathcal{T}}/\sqrt{Q_{\mathcal{T}}}$.}
Let $S_{\mathcal{T}} = \sum_{t\in\mathcal{T}} \xi_t = n_{\mathrm{eff}}(\widehat\theta - \theta_0)$. We apply Theorem~\ref{thm:selfnorm} (stated in Appendix~\ref{app:martingale}). Assumption~\ref{ass:var_growth} implies $V_{\mathcal{T}}^2 \to \infty$ in probability.

To show $Q_{\mathcal{T}}/V_{\mathcal{T}}^2\pto 1$, define the martingale differences
$\Delta_t:=\xi_t^2-\E[\xi_t^2\mid \mathcal{F}_{t-1}]$ and the partial sums
$M_k:=\sum_{t=1}^k\Delta_t$ for $k=1,\dots,n$. Then $\{M_k,\mathcal{F}_k\}$ is a martingale and
$M_n=Q_{\mathcal{T}}-V_{\mathcal{T}}^2$.
By Lemma~\ref{lem:moment_bounds}, $\sup_t\E[\xi_t^4]<\infty$, and Jensen implies
$\E[\E[\xi_t^2\mid\mathcal{F}_{t-1}]^2]\le \E[\xi_t^4]$, hence
$\sup_t\E[\Delta_t^2]\le 4\sup_t\E[\xi_t^4]<\infty$.
By orthogonality of martingale differences,
\[
\E[M_n^2]=\sum_{t\in\mathcal{T}}\E[\Delta_t^2]=O(n_{\mathrm{eff}}).
\]
Let $\mathcal{E}_{\mathcal{T}}:=\{V_{\mathcal{T}}^2\ge v_- n_{\mathrm{eff}}\}$,
which satisfies $\Pp(\mathcal{E}_{\mathcal{T}})\to 1$ by Assumption~\ref{ass:var_growth}.
On $\mathcal{E}_{\mathcal{T}}$,
\[
\left|\frac{Q_{\mathcal{T}}}{V_{\mathcal{T}}^2}-1\right|
\le \frac{|M_n|}{v_- n_{\mathrm{eff}}}.
\]
Therefore, for any $\eta>0$,
\[
\Pp\!\left(\left|\frac{Q_{\mathcal{T}}}{V_{\mathcal{T}}^2}-1\right|>\eta\right)
\le \Pp(\mathcal{E}_{\mathcal{T}}^c)
+\frac{\E[M_n^2]}{\eta^2 v_-^2 n_{\mathrm{eff}}^{2}}
\to 0,
\]
so $Q_{\mathcal{T}}/V_{\mathcal{T}}^2\pto 1$.

For the Lyapunov condition with $\delta=2$, Lemma~\ref{lem:moment_bounds} gives $\sup_{t\in\mathcal{T}}\E[|\widehat\phi_t|^4]<\infty$; since $\xi_t=T_t(\widehat\phi_t-\theta_0)$ and $T_t\in\{0,1\}$, it follows that $\sup_{t\le n}\E[|\xi_t|^4]<\infty$ for each $n$. Jensen then implies $\sup_{t\le n}\E[|\xi_t|^{2+\delta}]<\infty$ for any $\delta\in(0,2]$.
\[
\frac{\sum_{t\in\mathcal{T}}\E|\xi_t|^4}{V_{\mathcal{T}}^4}
\le \frac{C n_{\mathrm{eff}}}{(v_- n_{\mathrm{eff}})^2}\to 0,
\]
hence condition (iii) of Theorem~\ref{thm:selfnorm} holds.
By Theorem~\ref{thm:selfnorm}, $S_{\mathcal{T}}/\sqrt{Q_{\mathcal{T}}} \dto \mathcal{N}(0,1)$.

\medskip\noindent\textbf{Step 2: Replace $Q_{\mathcal{T}}$ with $(n_{\mathrm{eff}}-1)\widehat{V}$.}
We have
\begin{align*}
(n_{\mathrm{eff}}-1)\widehat V
&= \sum_{t\in\mathcal{T}} (\widehat{\phi}_t - \widehat\theta)^2\\
&= \sum_{t\in\mathcal{T}} (\xi_t - (\widehat\theta - \theta_0))^2\\
&= Q_{\mathcal{T}} - n_{\mathrm{eff}}(\widehat\theta - \theta_0)^2.
\end{align*}
Since $S_{\mathcal{T}}/\sqrt{Q_{\mathcal{T}}} = O_\Pp(1)$, we have $S_{\mathcal{T}}^2/Q_{\mathcal{T}} = O_\Pp(1)$, hence
\begin{align*}
\frac{n_{\mathrm{eff}}(\widehat\theta - \theta_0)^2}{Q_{\mathcal{T}}}
&= \frac{S_{\mathcal{T}}^2}{n_{\mathrm{eff}} Q_{\mathcal{T}}}\\
&= \frac{1}{n_{\mathrm{eff}}} \cdot \frac{S_{\mathcal{T}}^2}{Q_{\mathcal{T}}}
= o_\Pp(1).
\end{align*}
Thus $(n_{\mathrm{eff}}-1)\widehat V/Q_{\mathcal{T}} \pto 1$, and by Slutsky's lemma:
\[
\frac{S_{\mathcal{T}}}{\sqrt{(n_{\mathrm{eff}}-1)\widehat V}} = \frac{S_{\mathcal{T}}}{\sqrt{Q_{\mathcal{T}}}} \cdot \sqrt{\frac{Q_{\mathcal{T}}}{(n_{\mathrm{eff}}-1)\widehat V}} \dto \mathcal{N}(0,1). \qedhere
\]
\end{proof}

\section{Additional results}\label{app:additional_results}

This appendix collects results on variance stabilization, oracle equivalence,
and variance-growth primitives.

\subsection{Variance Consistency Under Stabilization}\label{sec:variance_stab}

When the conditional variance stabilizes, we obtain a conventional CLT.

\begin{assumption}[Variance stabilization]\label{ass:stab}
There exists $V \in (0,\infty)$ such that $V_{\mathcal{T}}^2/n_{\mathrm{eff}} \pto V$.
\end{assumption}

\begin{corollary}[Non-studentized CLT under variance stabilization]\label{cor:nonstudentized}
Assume Assumptions~\ref{ass:iid}, \ref{ass:logging_integrity}, \ref{ass:unconf}, \ref{ass:overlap}, \ref{ass:predictable_nuis}, \ref{ass:T_det}, \ref{ass:moments}, \ref{ass:nuis_stab}, \ref{ass:var_growth}, and \ref{ass:stab}. Then $\sqrt{n_{\mathrm{eff}}}(\widehat{\theta}-\theta_0)\Rightarrow \mathcal{N}(0,V)$ and $\widehat V\to_p V$.
\end{corollary}

\begin{proof}
By Theorem~\ref{thm:studentized},
\[
\frac{\sqrt{n_{\mathrm{eff}}}(\widehat\theta-\theta_0)}{\sqrt{\widehat V}} \dto \mathcal{N}(0,1).
\]
Moreover the proof of Theorem~\ref{thm:studentized} shows both
$Q_{\mathcal{T}}/V_{\mathcal{T}}^2\pto 1$ and
$(n_{\mathrm{eff}}-1)\widehat V/Q_{\mathcal{T}}\pto 1$, hence
$\widehat V = (V_{\mathcal{T}}^2/n_{\mathrm{eff}})\cdot(1+o_{\Pp}(1))$.
Under Assumption~\ref{ass:stab}, $V_{\mathcal{T}}^2/n_{\mathrm{eff}}\pto V$,
so $\widehat V\pto V$ and Slutsky yields
$\sqrt{n_{\mathrm{eff}}}(\widehat\theta-\theta_0)\dto \mathcal{N}(0,V)$.
\end{proof}

\subsection{Oracle Equivalence}\label{sec:oracle_equiv}

Under additional nuisance consistency, the feasible estimator is asymptotically equivalent to the oracle.

\begin{assumption}[$L^2$ nuisance consistency]\label{ass:l2}
Let $m_a^\star(x):=\E[Y(a)\mid X=x]$. Then, for each $a\in\{0,1\}$, we have (where the expectation is over the full data stream and any learner randomness used to construct $\widehat m$)
\[
\frac{1}{n_{\mathrm{eff}}}\sum_{t\in\mathcal{T}}
\E\Big[
\big(\widehat m_{t-1,a}(X_t)-m_a^\star(X_t)\big)^2
\Big]\to 0.
\]
\end{assumption}

\begin{theorem}[Oracle equivalence]\label{thm:oracle_equivalence}
Assume Assumptions~\ref{ass:iid}, \ref{ass:logging_integrity}, \ref{ass:unconf}, \ref{ass:overlap}, \ref{ass:predictable_nuis}, and \ref{ass:l2}. Assume also the scored set $\mathcal{T}$ is deterministic (Assumption~\ref{ass:T_det}). Let $m^\star:=(m_0^\star,m_1^\star)$ and define the oracle estimator
\[
\widehat{\theta}^\star := \frac{1}{n_{\mathrm{eff}}}\sum_{t\in\mathcal{T}} \phi_t(m^\star).
\]
Then
\[
\sqrt{n_{\mathrm{eff}}}(\widehat{\theta} - \widehat{\theta}^\star) \pto 0.
\]
\end{theorem}

\begin{proof}
Let $e_{t-1,a}(x):=\widehat m_{t-1,a}(x)-m_a^\star(x)$.
A direct algebraic calculation gives, for each scored index $t\in\mathcal{T}$,
\[
\phi_t(\widehat m_{t-1})-\phi_t(m^\star)
=
-(A_t-\pi_t)\left(\frac{e_{t-1,1}(X_t)}{\pi_t}+\frac{e_{t-1,0}(X_t)}{1-\pi_t}\right)
=:\Delta_t.
\]
By predictability, $(e_{t-1,0},e_{t-1,1})$ is $\cF_{t-1}$-measurable, and $\pi_t$ is $\cG_t$-measurable.
Using iterated expectations,
\[
\E[\Delta_t\mid \cF_{t-1}]
=
\E\!\left[\E[\Delta_t\mid \cG_t]\mid \cF_{t-1}\right]
=
-\E\!\left[\left(\frac{e_{t-1,1}(X_t)}{\pi_t}+\frac{e_{t-1,0}(X_t)}{1-\pi_t}\right)
\E[A_t-\pi_t\mid \cG_t]\Bigm|\cF_{t-1}\right]
=0,
\]
so $\{\Delta_t,\cF_t\}$ is a martingale difference sequence over $t\in\mathcal{T}$.
Now
\[
\sqrt{n_{\mathrm{eff}}}\,(\widehat\theta-\widehat\theta^\star)
=
\frac{1}{\sqrt{n_{\mathrm{eff}}}}\sum_{t\in\mathcal{T}} \Delta_t.
\]
By orthogonality of martingale differences,
\[
\E\!\left[\Big(\sqrt{n_{\mathrm{eff}}}\,(\widehat\theta-\widehat\theta^\star)\Big)^2\right]
=
\frac{1}{n_{\mathrm{eff}}}\sum_{t\in\mathcal{T}}\E[\Delta_t^2].
\]
Conditional on $\cG_t$, $\E[(A_t-\pi_t)^2\mid \cG_t]=\pi_t(1-\pi_t)$, hence
\[
\E[\Delta_t^2\mid \cG_t]
=
\pi_t(1-\pi_t)\left(\frac{e_{t-1,1}(X_t)}{\pi_t}+\frac{e_{t-1,0}(X_t)}{1-\pi_t}\right)^2
\le C\big(e_{t-1,1}(X_t)^2+e_{t-1,0}(X_t)^2\big)
\]
for a constant $C<\infty$ depending only on the overlap constant $\varepsilon$.
Taking expectations and averaging over $t\in\mathcal{T}$ yields
\[
\E\!\left[\Big(\sqrt{n_{\mathrm{eff}}}\,(\widehat\theta-\widehat\theta^\star)\Big)^2\right]
\le
\frac{C}{n_{\mathrm{eff}}}\sum_{t\in\mathcal{T}}
\E\big[e_{t-1,1}(X_t)^2+e_{t-1,0}(X_t)^2\big]
\to 0
\]
by Assumption~\ref{ass:l2}. Therefore $\sqrt{n_{\mathrm{eff}}}(\widehat\theta-\widehat\theta^\star)\pto 0$.
\end{proof}

\begin{remark}[Relation to i.i.d.\ DML rate conditions]\label{rem:n14}
In i.i.d.\ DML analyses of the ATE with \emph{unknown} propensities, asymptotic linearity typically requires a product-rate condition between the propensity-score error and the outcome-regression error (often enforced via $o_P(n^{-1/4})$-type rates for each component); see, e.g., \citep{chernozhukov2018double,chernozhukov2022locally}.
In our setting, propensities are assumed known and logged, so the remaining nuisance component enters only through predictable regression adjustments.
As a result, nuisance convergence is not needed for validity of the studentized CLT, and $L^2$ consistency is sufficient for oracle equivalence (Theorem~\ref{thm:oracle_equivalence}).
\end{remark}

\subsection{A Sufficient Condition for Variance Growth}\label{sec:var_growth_sufficient}

Assumption~\ref{ass:var_growth} is mild but abstract. The following provides a primitive sufficient condition.

\begin{assumption}[Nontrivial noise]\label{ass:noise}
$\E[\sigma_1^2(X) + \sigma_0^2(X)] \ge \underline{\sigma}^2 > 0$.
\end{assumption}

\begin{lemma}[A sufficient condition for variance growth]\label{lem:var_growth_suff}
Assume Assumptions~\ref{ass:iid}, \ref{ass:logging_integrity}, \ref{ass:unconf}, \ref{ass:overlap}, and \ref{ass:noise}. Let $\xi_t:=\phi_t(\widehat m_{t-1})-\theta_0$ with $\widehat m_{t-1}$ satisfying Assumption~\ref{ass:predictable_nuis}, and let $V_{\mathcal{T}}^2:=\sum_{t\in\mathcal{T}}\E[\xi_t^2\mid \cF_{t-1}]$ as in \eqref{eq:quadratic_variations}. Then there exists $v_->0$ such that for every scored set $\mathcal{T}$ with $n_{\mathrm{eff}}:=|\mathcal{T}|\to\infty$,
\[
\Pp\!\big(V_{\mathcal{T}}^2 \ge v_-\, n_{\mathrm{eff}}\big)\to 1,
\]
so Assumption~\ref{ass:var_growth} holds.
\end{lemma}

\begin{proof}
Fix $t\in\mathcal{T}$. By Proposition~\ref{prop:var_decomp} and nonnegativity of
the augmentation term,
\[
\E[\xi_t^2\mid \cG_t]
\ge
\frac{\sigma_1^2(X_t)}{\pit}+\frac{\sigma_0^2(X_t)}{1-\pit}.
\]
By overlap (Assumption~\ref{ass:overlap}), $\pi_t\le 1-\varepsilon$ and $1-\pi_t\le 1-\varepsilon$ (so $1/\pi_t\ge 1/(1-\varepsilon)$ and $1/(1-\pi_t)\ge 1/(1-\varepsilon)$), hence
\[
\frac{\sigma_1^2(X_t)}{\pi_t}+\frac{\sigma_0^2(X_t)}{1-\pi_t}
\ge
\frac{\sigma_1^2(X_t)+\sigma_0^2(X_t)}{1-\varepsilon}.
\]
Taking $\E[\cdot\mid \cF_{t-1}]$ and using i.i.d.\ arrivals (Assumption~\ref{ass:iid}) yields
\[
\E[\xi_t^2\mid \cF_{t-1}]
=
\E\big[\E[\xi_t^2\mid \cG_t]\mid \cF_{t-1}\big]
\ge
\frac{1}{1-\varepsilon}\E[\sigma_1^2(X)+\sigma_0^2(X)]
\ge
\frac{\underline{\sigma}^2}{1-\varepsilon}
=:v_-,
\]
almost surely.
Summing over $t\in\mathcal{T}$ gives $V_{\mathcal{T}}^2=\sum_{t\in\mathcal{T}}\E[\xi_t^2\mid \cF_{t-1}]\ge v_- n_{\mathrm{eff}}$ almost surely, hence the probability statement holds.
\end{proof}

\section{Assumption-to-Result Map}\label{app:assumptions}

\begin{table}[t]
\centering
\small
\begin{tabular}{@{}p{0.1\linewidth} p{0.4\linewidth} p{0.4\linewidth}@{}}
\toprule
Assumption & Informal description & Used in \\
\midrule
\ref{ass:sutva} & SUTVA / no interference & Identification of $\theta_0$ (Section~\ref{sec:model}); standing throughout. \\
\ref{ass:iid} & i.i.d.\ arrivals / no selection (or mean-stationarity alternative) & Lem.~\ref{lem:mds}; Prop.~\ref{prop:var_decomp}; Lem.~\ref{lem:var_growth_suff} (Appendix~\ref{sec:var_growth_sufficient}); Thm.~\ref{thm:studentized} via Lem.~\ref{lem:mds}. \\
\ref{ass:logging_integrity} & logged executed propensities & Lem.~\ref{lem:mds}; Prop.~\ref{prop:var_decomp}; Lem.~\ref{lem:var_growth_suff}. \\
\ref{ass:unconf} & sequential randomization / no unmeasured confounding & Lem.~\ref{lem:mds}; Prop.~\ref{prop:var_decomp}; Lem.~\ref{lem:var_growth_suff}. \\
\ref{ass:T_det} & prespecified (deterministic) scored set $\mathcal{T}$ & Thm.~\ref{thm:studentized} (via predictable indicator $\ind{t\in\mathcal{T}}$); Appendix~\ref{app:proofs} proof of Thm.~\ref{thm:studentized}. \\
\ref{ass:predictable_nuis} & nuisances used for unit $t$ depend only on past & Lem.~\ref{lem:mds}; Thm.~\ref{thm:studentized}. \\
\ref{ass:overlap} & executed overlap on scored units & Lem.~\ref{lem:moment_bounds}; Lem.~\ref{lem:var_growth_suff}; Thm.~\ref{thm:studentized} via Lem.~\ref{lem:moment_bounds}. \\
\ref{ass:moments} & bounded moments of potential outcomes & Lem.~\ref{lem:moment_bounds}; Thm.~\ref{thm:studentized}. \\
\ref{ass:nuis_stab} & stability of nuisance fits (bounded moments) & Lem.~\ref{lem:moment_bounds}; Thm.~\ref{thm:studentized}. \\
\ref{ass:var_growth} & variance growth ($V_{\mathcal{T}}^2\to\infty$) & Thm.~\ref{thm:studentized} (verifying conditions of Thm.~\ref{thm:selfnorm}). \\
\ref{ass:noise} & nondegenerate outcome noise & Cor.~\ref{cor:var_growth_suff}; Lem.~\ref{lem:var_growth_suff}. \\
\ref{ass:stab} & deterministic variance limit (stabilization) & Cor.~\ref{cor:nonstudentized}. \\
\bottomrule
\end{tabular}
\caption{Assumption-to-result map}\label{tab:assumption_ledger}
\end{table}

\section{Data contract and operational logging requirements}\label{app:repro}

This appendix describes a minimal ``data contract'' for implementing Theorem~\ref{thm:studentized} in an experimentation platform. The contract separates (i) auditable pipeline conditions (executed-propensity logging and predictable nuisance construction) from (ii) substantive causal/model assumptions (e.g., SUTVA/no interference, sequential randomization, and no selection) that are not fully verifiable from logs alone.

\begin{table}[t]
\centering
\small
\begin{tabular}{@{}p{0.46\linewidth} p{0.48\linewidth}@{}}
\toprule
Persisted artifact & Audit / validation check \\
\midrule
Executed propensity $\pi_t$ and assignment record for each unit $t$ (including any clipping rule and RNG call metadata) &
\textbf{Executed propensity integrity.} Verify that the logged $\pi_t$ is the probability passed to the randomization device. As a coarse calibration diagnostic, compare $\overline A_B$ to $\overline\pi_B$ within propensity bins $B$; this can catch severe implementation errors, but it does not certify correct logging/randomization or validate the causal assumptions (Remark~\ref{rem:logging_audit}). \\[0.4em]
Experiment-level overlap rule (e.g., clipping threshold $\varepsilon$) and the realized min/max of $\pi_t$ on scored units &
\textbf{Executed overlap.} Check that $\pi_t\in[\varepsilon,1-\varepsilon]$ for all $t\in\mathcal{T}$. If overlap is enforced by clipping, confirm clipping occurred at assignment time. \\[0.4em]
Scored-set definition and indicator $T_t:=\ind{t\in\mathcal{T}}$ (or a documented predictable rule) &
\textbf{Scored-set integrity.} Confirm $\mathcal{T}$ is prespecified (Assumption~\ref{ass:T_det}) or, if a predictable rule is used (e.g., due to outcome delay), document the rule and verify it does not depend on future outcomes. \\[0.4em]
Nuisance-fit artifacts for each scoring block (training indices, timestamps, learner configuration, hyperparameters, and randomness/seeds) &
\textbf{Predictability.} Verify that the nuisance fit used to score unit $t$ was trained only on indices $<t$, and that the nuisance is held fixed while scoring all units in its associated block (Definition~\ref{def:forward} and Assumption~\ref{ass:predictable_nuis}). \\[0.4em]
Prespecified reporting horizon $n$ (and, if applicable, an analysis-query log) &
\textbf{Fixed horizon / no optional stopping.} Confirm that confidence intervals are produced only at the prespecified horizon(s). If the experiment is continuously monitored or stopped adaptively, the fixed-horizon interval is not guaranteed; use time-uniform methods instead (Remark~\ref{rem:optional_stopping}). \\
\bottomrule
\end{tabular}
\caption{Minimal logging contract and audit checks for fixed-horizon inference.}\label{tab:data_contract}
\end{table}

\paragraph{Design-time overlap vs analysis-time truncation.}
Overlap should be enforced at the assignment stage if needed (e.g., by clipping $\pi_t$ before randomization). Post-hoc truncation of inverse-propensity weights in the analysis generally changes the estimand and can introduce bias; if units are dropped or trimmed based on realized propensities, the target estimand must be redefined explicitly.

\paragraph{Extension to multi-arm treatments.}
For $K>2$ arms with assignment probabilities $\pi_{t,a}$, the same contract applies: log the executed $\pi_{t,a}$ for all arms, enforce $\pi_{t,a}\ge \varepsilon$ on scored units, and construct nuisance fits predictably. The AIPW score and studentizer are obtained by replacing binary inverse-propensity terms by their $K$-arm analogues.

% ============================================================================
% REFERENCES
% ============================================================================

\end{document}